\documentclass[11pt, letterpaper,notitlepage]{article}
\usepackage[letterpaper, total={6.5in, 9in}]{geometry}
\usepackage{enumitem, tikz}
\usepackage{booktabs} % For formal tables
\usepackage[colorinlistoftodos,textsize=tiny,textwidth=2cm,color=red!25!white,obeyFinal]{todonotes}
\usepackage{verbatim}
\usepackage{hyperref} %Links within doc
\hypersetup{
    colorlinks=true,
    allcolors=blue,
}

\usepackage{float}
\usepackage{amsmath} 
\usepackage{amsthm}
\usepackage{amssymb}
\usepackage[ruled]{algorithm2e} % For algorithms

\usepackage[normalem]{ulem}
\SetAlFnt{\small}
\SetAlCapFnt{\small}
\SetAlCapNameFnt{\small}
\SetAlCapHSkip{0pt}
\newcommand{\initOneLiners}{%
    \setlength{\itemsep}{0pt}
    \setlength{\parsep }{0pt}
    \setlength{\topsep }{0pt}
}

\renewcommand{\ss}{\varsigma}

\usepackage{microtype, todonotes}%if unwanted, comment out or use option "draft"

\usepackage{url}

\newcommand{\TS}{\texttt{TreeSearch}}
\newcommand{\TM}{\texttt{TreeMatch}}

\newcommand{\Ac}{\mathcal{A}}
\newcommand{\Dc}{\mathcal{D}}
\newcommand{\Pc}{\mathcal{P}}

\newcommand{\Xc}{\mathcal{X}}
\newcommand{\Yc}{\mathcal{Y}}

\newcommand{\R}{\mathbb{R}}

\newcommand{\HTB}{\texttt{GroveBuild}}
\newcommand{\HTM}{\texttt{GroveMatch}}

\DeclareMathOperator*{\argmin}{arg\,min}

\newcommand{\Prob}[1]{\textbf{Pr}\left[{#1}\right]}
\newcommand{\Exp}[1]{\textbf{E}\left[{#1}\right]}

\usetikzlibrary{shapes.geometric}
\usetikzlibrary{calc,scopes}
\usetikzlibrary{intersections}
\usetikzlibrary{shapes,positioning}
\renewcommand{\P}{\mathcal{P}}
\newcommand{\A}{\mathcal{A}}
\newcommand{\B}{\mathcal{B}}

\newtheorem{theorem}{Theorem}
\newtheorem{lemma}[theorem]{Lemma}
\newtheorem{observation}[theorem]{Observation}
\newtheorem{definition}[theorem]{Definition}
\newtheorem{corollary}[theorem]{Corollary}

\title{Competitively Pricing Parking in a Tree}

\begin{document}

\author{Max Bender \\ Computer Science Department\\University of Pittsburgh\\\texttt{MCB121@pitt.edu}\and Jacob Gilbert\\ Computer Science Department\\University of Pittsburgh\\\texttt{JMG264@pitt.edu} \and Aditya Krishnan\thanks{This work was done in part while this author was a student at Carnegie Mellon University}\\ Computer Science Department\\Johns Hopkins University\\\texttt{aditya.krishnan94@gmail.com} \and Kirk Pruhs\thanks{Supported in part by NSF grants CCF-1421508 and CCF-1535755, and an IBM Faculty Award.}\\ Computer Science Department\\University of Pittsburgh\\\texttt{kirk@cs.pitt.edu}}

\maketitle

\begin{abstract}
Motivated by demand-responsive parking pricing systems
we consider posted-price algorithms for the online metrical matching problem and the online metrical searching problem
in a tree metric. Our main result is a poly-log competitive
posted-price algorithm for online metrical searching.
\end{abstract}

 \section{Introduction}
 
Since 2011 SFpark has been San Francisco's system for managing the availability of on-street parking~\cite{SFparkWikipedia,SFparkAccess,SFparkHomepage}.
The goal of the system is to reduce the time and fuel wasted by drivers searching for an open space.
The system monitors parking usages via sensors embedded in the pavement and distributes this information in real-time to drivers
via SFpark.org and phone apps. 
SFpark periodically adjusts parking meter pricing to manage demand, to lower prices in underutilized areas, and to raise prices in overutilized areas. 
Prices can range from a minimum of 25 cents to a maximum of 7 dollar per hour during normal hours with a 18 dollars per hour cap for special events such as baseball games or street fairs. 
Several other cities in the world have similar demand-responsive parking pricing
systems, for example Calgary has had the ParkPlus system since 2008~\cite{calgaryparkplus}.

The problem of centrally assigning drivers to parking spots to minimize time and  fuel usage is naturally modeled by the online metrical matching problem. 
The setting for online metrical matching consists of a collection of $k$ servers (the parking spots) located at various
locations within a metric space. The  algorithm
then sees an online sequence of requests over time
that arrive at various locations in the metric space (the drivers arriving to look for a parking spot). In response to a request, the online algorithm
must match the request (car) to some server (parking spot) that has
not been previously matched; Conceptually we interpret
this matching as the request (car) moving to  the
location of the matched server (parking spot).  The objective goal 
is to minimize the aggregate distance traveled by the 
requests (cars).

We also consider what we call the online metrical search problem, which is an important special case of the online metrical matching problem.  This is a promise problem in that the adversary is constrained to guarantee that there is an
optimal matching for which only one edge has positive cost. 
It is useful to conceptually think 
of online metrical search as the following parking problem: the setting consists of many parking spots at various locations
in a metric space and a single car that is initially parked at some location in the metric space.
Over time the parking spots are decommissioned one by one until only one parking spot is left in commission. If at any time the car is not parked at an in-commission parking spot, then the car must move to a parking spot that is still in commission. The objective is to minimize the aggregate distance traveled by the car. The optimal solution is to move the car directly to the  
last remaining parking spot.

The online
metrical search problem is a special case of the online
metrical matching problem because the parking spots can
be viewed as servers and the decommissioning of
a parking spot can be simulated by the arrival of a 
request at the location of that parking spot. So a lower bound on the
competitive ratio for the online metrical search problem for
a particular metric space also gives a lower bound
for the online metrical matching problem on the metric space. Conversely it seems that in terms of the optimal
competitive ratio, online metric search is no easier than
metric matching. In particular, there is no known example of a metric 
space where the optimal competitive ratio for online metrical matching is known to be significantly greater than the optimal competitive ratio for 
online metrical search on that metric space. 
For example on a line metric, the online metrical search
problem is better known as the  ``cow path problem'', and the optimal deterministic competitive ratio is known to be 9~\cite{Borodin}, while the best known lower bound on the deterministic
competitive ratio for online metrical matching
on a line metric is 9.001~\cite{FuchsHK05},
worse only by a minuscule factor.

In order to be implementable within the context of SFpark, online algorithms must be posted-price algorithms. In this setting, posted-price means that 
before each request arrives, the online algorithm
sets a price on each unused server (parking spot) without knowing the location where the next request will arrive. Furthermore,
each  request is assumed to be a selfish agent who moves
to the available server (parking spot) that minimizes the sum of 
the price of and distance to that server.  
The objective remains to minimize the aggregate distance traveled by the requests. So conceptually the objective of the parking pricing agency is minimizing social cost, not maximizing revenue. 

Research into posted-price algorithms for online metrical matching was initiated in \cite{CohenEFJ15} as part of a line of research to study the use of posted-price 
algorithms to minimize social cost in online optimization problems. 
%There are now a handful of papers the line of research initiated by \cite{CohenEFJ15}, which can be partitioned into those dealing
%with server problems \cite{CohenEFJ15} such as metrical matching, and those dealing with scheduling problems 
%\cite{FeldmanFR17,ImMPS17}.
As a posted-price algorithm is a valid online algorithm, one can not expect to obtain a better competitive ratio for posted-price algorithms than what is achievable
by online algorithms. So this research line has primarily focused on problems where
the optimal competitive ratio achievable by an online algorithm is (perhaps approximately) known
and seeks to determine whether a similar competitive ratio can be
(again perhaps approximately) achieved by a posted-price algorithm. The higher level goal is to determine the
increase in social cost that is necessitated by the restriction that an algorithm
has to use posted prices to incentivize selfish agents, instead of being able to mandate agent behavior. 

An $O(\log \Delta)$-competitive randomized posted-price algorithm for metric
matching on a line metric is given in  \cite{CohenEFJ15}
where $\Delta$ is  the ratio of the distance  between the furthest two servers and
the distance between the closest two servers. 
No $o(\log k)$-competitive (not necessarily posted-price) algorithm is known 
for online metric matching on a line metric. 
So arguably, on a line
metric there is a 
posted-price algorithm that is nearly as competitive
as the best known centralized
online algorithm. 

Our original research goal was to determine whether 
posted-price algorithms can be similarly competitive
 with a centralized online algorithm for tree metrics for online metrical matching. In order to be more specific about our goal,
 we need to review a bit. A tree metric is represented by a  tree  $T = (V, E)$ with positive real edge weights where the distance $d_T(u, v)$ between vertices $u, v \in V$ is  the shortest path between vertices $u$ and $v$ in $T$.
 There is a deterministic online algorithm that is $(2k-1)$-competitive for online metric matching in
any metric space, and no deterministic online algorithm can achieve a better competitive ratio for online metric searching in a tree metric~\cite{KalyanasundaramP93,KhullerMV94}.
An $O(\log k)$-competitive randomized algorithm for
online metric matching in $O(\log k)$-HST's (Hierarchically Separated Trees) is given in \cite{MeyersonNP06}. By combining
this result with results about randomly embedding metric spaces into HST's \cite{Bartal96,Bartal98,FakcharoenpholRT04},
\cite{MeyersonNP06} obtained an $O(\log^3 k)$-competitive randomized algorithm for online metric matching in a general metric space. 
Following this general approach, \cite{BansalBGN14} later obtained an $O(\log^2 k)$-competitive
randomized  algorithm for online metrical search in an arbitrary metric by giving an $O(\log k)$-competitive randomized algorithm 
for $2$-HST's. No better results are known for tree
metrics, so all evidence points to tree metrics as being
as hard as general metrics for online metrical matching. Thus, more specifically
our original research goal was to determine whether there 
is poly-log competitive randomized posted-price
algorithm for the online metrical matching problem
on a tree metric. 
Before stating our progress toward this goal, it will be useful to review  the literature a bit more.

\subsection{Prior Related Work}
The most obvious algorithmic design approach for posted-price problems
is to directly design 
a pricing algorithm from scratch, as is done for metrical task systems in \cite{CohenEFJ15}, but this is not the most common approach in the literature. 
Two less direct  algorithmic design paradigms have emerged in the literature. The first algorithmic design paradigm is what we will call {\em mimicry}. A posted-price algorithm $A$ {\em mimics}
an online algorithm $B$ if the probability that $B$ will take a particular action is equal the the 
probability that a self-interested
agent will choose this same action when the prices of actions are set using $A$. For example,
\cite{CohenEFJ15}  shows how to set prices to mimic the $O(\log \Delta)$-competitive Harmonic algorithm
for online metric matching on a line metric from \cite{GuptaL12}. As another example,   \cite{FeldmanFR17} shows how to set prices to mimic the $O(1)$-competitive  algorithm Slow-Fit from \cite{Aspnes1997,AzarKPPW97} for the problem of minimizing makespan on related machines. 
However, for some problems it is not possible to mimic known  competitive algorithms using posted prices. For such problems, another algorithmic design paradigm is what we will 
call {\em monotonization}. In the monotonization algorithm design approach, one first seeks to characterize the online algorithms
that can be mimicked, and then designs such an online algorithm. In the  examples in the literature, this characterization 
involves some sort of monotonicity property. 
For example, monotonization is used in \cite{CohenEFJ15} to obtain an $O(k)$-competitive posted-price algorithm for the $k$-server problem on a line metric, and in \cite{CohenEFJ19}
to to obtain an $O(k)$-competitive posted-price algorithm for the $k$-server problem on a tree metric. 
Since no deterministic algorithm can be better than 
$k$-competitive for the $k$-server problem in any metric~\cite{k-server-lower}, 
this shows that in these settings, there is minimal increase
in social cost necessitated by the use of posted-prices.
As another example, monotonization is used in \cite{ImMPS17} to obtain an $O(1)$-competitive posted-price algorithm for minimizing maximum flow time on related machines.

For online metric matching on a line metric, better competitive ratios are achievable. An $O(k^{.59})$-competitive deterministic 
online algorithm was given in \cite{AntoniadisBNPS14}.
Subsequently several different $O(\log n)$-competitive randomized online algorithms for a line are given in \cite{GuptaL12}; these algorithms
leverage special properties of HST's constructed from a line metric. As already mentioned, 
\cite{GuptaL12} also showed that the natural Harmonic algorithm is $O(\log \Delta)$-competitive.
An $O(\log^2 k)$-competitive deterministic online algorithm was given in \cite{NayyarR17},
and this was later improved to $O(\log k)$ in
\cite{Raghvendra18}.
Super-constant lower bounds for various types of algorithms are given in \cite{AntoniadisFT18,KoutsoupiasN03}.
More generally, the algorithm for online
metric matching given in \cite{NayyarR17} has the property that for every metric space, its competitive
ratio is at most $O(\log^2 k)$ times the optimal competitive
ratio achievable by any deterministic algorithm on that metric space.

\subsection{Our Contribution}

There is no hope to  mimic any of the online algorithms for online metrical matching that are based on
HST's as HST's by their very nature lose too much information about the structure of a tree metric.
Therefore 
we adopt the monotonization approach. 
In Section \ref {sec:price} we identify a monotonicity property that characterizes mimicable  algorithms for
online metrical matching in tree metrics.
Roughly speaking this property says that if 
a request were to have arrived on the route to its
desired server, then the probability that the
request would still have been matched to this server can not 
decrease.
Thus we reduce finding a post-priced algorithm
to finding a monotone algorithm.

In Section \ref{sect:search} we give
an algorithm \TS\ for the
online metrical search problem on a tree metric.
The algorithm is based on the classic multiplicative weights algorithm for
online learning from experts \cite{arora2012multiplicative}.
Conceptually there is one expert $E^\ell$ for each
leaf $\ell$ of the tree $T$. Expert $E^\ell$
always recommends that the car/request
travels toward the leaf $\ell$. 
Thus expert $E^\ell$ pays a cost of one
whenever a parking spot on the path from
the root to $\ell$ is decommissioned,
a cost of zero when other parking spots
are decommissioned, and an infinite cost
if there are no remaining parking spots
on the path from the root to $\ell$. Let $\pi^\ell_t$ 
be the probability that the multiplicative
weights algorithm has associated with
expert $E^\ell$ right before request $r_t$
arrives. 
Let $v^\ell_{t}$ be the location of
the car just before request $r_t$ arrives
if the advice of expert $E^\ell$  had always been
followed. The algorithm \TS\ maintains
the invariant that right
before request $r_t$ arrives, the 
probability that the car is at a vertex $v$
is $\sum_{\ell : v^\ell_t =v} \pi^\ell_t$,
the sum of the probabilities of the experts that
recommend that the car should be parked at $v$. 
The most technically difficult part of the
 algorithm design process was maintaining 
this invariant. We then upper bound  the expected
number of jumps made by the \TS\ algorithm,
where a jump is a  movement of the car by a positive amount. 
Finally, we show how to extend \TS\ to be a
monotone algorithm \TM\ for online metrical matching on a tree metric. 

In Section algorithm for 
online metric searching on a tree metric.
Before any requests arrive, an algorithm \HTB\ 
embeds the tree metric into what we will call a grove, which is a refinement of an HST that retains more information about the topology of the original metric space. It is probably easiest
to explain what a grove is by explaining the
difference in how one is constructed in 
comparison to how an HST is constructed.
The construction of each starts with a Low Diameter Decomposition (LDD) of the metric space. 
A LDD is a partition ${\mathcal P} = \{P_1, \ldots, P_n \}$ of 
the vertices of the metric space where each
part is connected and the
diameter of each part is an $\alpha$ factor smaller than the diameter of the whole metric space. The top of the HST consists
of a star where the center of the star is
the root of the HST, and there is one child
of the root for each part $P_i$.
In contrast, the top of a grove consists 
of the tree that remains after collapsing
each part to a single vertex. 
For both an HST and a grove, the construction then
proceeds recursively on each part. 
So intuitively the key difference is that
groves retain information
about the distances between parts in the LDD that the HST instead discards. 
See Figure 1 for a comparison of an HST and a grove constructed from the same LDD.

We then give a monotone algorithm \HTM\ for online metrical matching on a tree metric that utilizes the
algorithm \TM\ on each tree in the grove constructed from the tree metric.
We show that \HTM\ is poly-log competitive 
(more precisely $O(\log^6 \Delta \log^2 n)$-competitive) on metric search instances
 by induction on the levels of the grove. This  is an extension of a
similar induction argument in \cite{MeyersonNP06} that shows that a $O(\log n)$-competitive algorithm for a star (or a complete unit metric) can be extended to an algorithm
for a $O(\log n)$-HST with the loss of a poly-log factor in the competitiveness. 
However, our situation is complicated by the
fact the possible ways that a request can potentially move within a grove is  more complicated than the possible ways a request can move within an HST, and thus the induction is more complicated as the induction depends on when the request is moving ``up'' and when the request is moving ``down'' in trees within the grove. 
The  bound on the number of
jumps made by \TS\ translates  to a bound on the 
number of recursive calls made by \HTM. 
There is not a lot of wiggle room in our analysis,
and thus both the algorithm design and algorithm
analysis process are necessarily quite delicate. 
For example, if \TS\ made just 1\% 
more jumps than the bound that we can show, then the resulting competitiveness of \HTM\ would not be poly-logarithmic.
One consequence of this delicateness is that we can
not use  a black box LDD construction to build our grove, we need to construct our LDD in a  way
that tightly controls the variance of random properties of our grove.

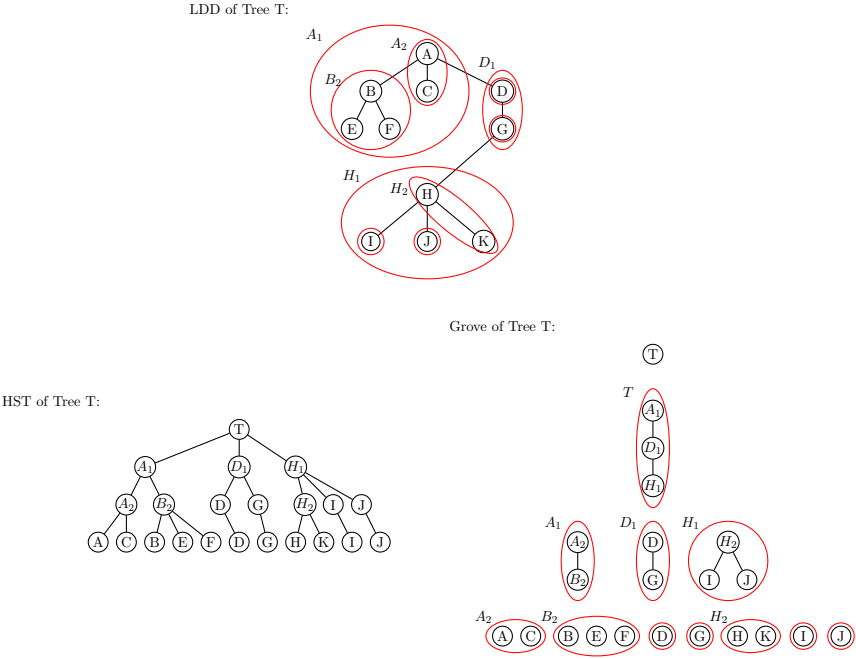
\begin{figure}[H]
\label{figure:HTexample}
\begin{center}
\centering
\begin{tikzpicture}[scale=.5, every node/.style={scale=0.5}]
\tikzstyle{vertex}=[circle,draw = black,minimum size=5pt, inner sep=2pt]
\node[vertex] (a) at (0, 0) {A};
\node[vertex] (b) at (-1.5, -1) {B};
\node[vertex] (c) at (0, -1) {C};
\node[vertex] (d) at (2, -1) { D};
\node[vertex] (e) at (-2, -2) {E};
\node[vertex] (f) at (-1, -2) {F};
\node[vertex] (g) at (2, -2) {G};
\node[vertex] (h) at (0, -3.75) {H};
\node[vertex] (i) at (-1.5, -5) {I};
\node[vertex] (j) at (0, -5) {J};
\node[vertex] (k) at (1.5, -5) {K};
\draw (a) -- (b) {};
\draw (a) -- (c) {};
\draw (a) -- (d) {};
\draw (b) -- (e) {};
\draw (b) -- (f) {};
\draw (d) -- (g) {};
\draw (g) -- (h) {};
\draw (h) -- (i) {};
\draw (h) -- (j) {};
\draw (h) -- (k) {};
\node[ellipse, draw=red, minimum height = 100pt, minimum width = 120pt] at (-1, -1) {};
\node[ellipse, draw=red, minimum height = 85pt, minimum width = 130pt] at (0, -4.5) {};
\node at (-3, .5) {$A_1$};
\node at (-2, -3.25) {$H_1$};
\node at (1.6, -.25) {$D_1$};
\node[ellipse, draw=red, minimum height = 60pt, minimum width = 30pt] at (2, -1.5) {};
\node[circle, draw=red, minimum width = 60pt] at (-1.5, -1.5) {};
\node[ellipse, draw=red, minimum height = 50pt, minimum width = 30pt] at (0, -.5) {};
\node at (-2.5, -.7) {$B_2$};
\node at (-.75, .25) {$A_2$};
\node[ellipse, rotate around={50:(0, 0)}, draw=red, minimum height = 85pt, minimum width = 25pt] at (.7, -4.3) {};
\node at (-.75, -3.6) {$H_2$};
\node[circle, draw=red, minimum width = 20pt] at (0, -5) {};
\node[circle, draw=red, minimum width = 20pt] at (-1.5, -5) {};
\node[circle, draw=red, minimum width = 20pt] at (2, -1) {};
\node[circle, draw=red, minimum width = 20pt] at (2, -2) {};
\node at (-5, 1.2) {LDD of Tree T:};

\tikzstyle{hstvertex}=[circle,draw=black,minimum size=15pt, inner sep=0pt]
\node[hstvertex] (t) at (-5, -10) {T};
\node[hstvertex] (tA1) at (-7.5, -11) {$A_1$};
\node[hstvertex] (tD1) at (-5, -11) {$D_1$};
\node[hstvertex] (tH1) at (-3.5, -11) {$H_1$};
\node[hstvertex] (tA2) at (-8, -12) {$A_2$};
\node[hstvertex] (tB2) at (-7, -12) {$B_2$};
\node[hstvertex] (tD2) at (-5.5, -12) {D};
\node[hstvertex] (tG2) at (-4.5, -12) {G};
\node[hstvertex] (tH2) at (-3.25, -12) {$H_2$};
\node[hstvertex] (tI) at (-2.5, -12) {I};
\node[hstvertex] (tJ) at (-1.75, -12) {J};
\node[hstvertex] (tA3) at (-8.75, -13) {A};
\node[hstvertex] (tC3) at (-8, -13) {C};
\node[hstvertex] (tB3) at (-7.25, -13) {B};
\node[hstvertex] (tE3) at (-6.5, -13) {E};
\node[hstvertex] (tF3) at (-5.75, -13) {F};
\node[hstvertex] (tD3) at (-5, -13) {D};
\node[hstvertex] (tG3) at (-4.25, -13) {G};
\node[hstvertex] (tH3) at (-3.5, -13) {H};
\node[hstvertex] (tK3) at (-2.75, -13) {K};
\node[hstvertex] (tI3) at (-2, -13) {I};
\node[hstvertex] (tJ3) at (-1.25, -13) {J};
\draw (t) -- (tA1) {};
\draw (t) -- (tH1) {};
\draw (t) -- (tD1) {};
\draw (tA1) -- (tA2) {};
\draw (tA1) -- (tB2) {};
\draw (tD1) -- (tD2) {};
\draw (tD1) -- (tG2) {};
\draw (tH1) -- (tH2) {};
\draw (tH1) -- (tI) {};
\draw (tH1) -- (tJ) {};
\draw (tA2) -- (tA3) {};
\draw (tA2) -- (tC3) {};
\draw (tB2) -- (tB3) {};
\draw (tB2) -- (tE3) {};
\draw (tB2) -- (tF3) {};
\draw (tD2) -- (tD3) {};
\draw (tG2) -- (tG3) {};
\draw (tH2) -- (tH3) {};
\draw (tH2) -- (tK3) {};
\draw (tI) -- (tI3) {};
\draw (tJ) -- (tJ3) {};
\node at (-10, -9.25) {HST of Tree T:};

\node[hstvertex] (ht) at (6, -8) {T};
\node[hstvertex] (hA1) at (6, -9.5) {$A_1$};
\node[hstvertex] (hD1) at (6, -10.5) {$D_1$};
\node[hstvertex] (hH1) at (6, -11.5) {$H_1$};
\node[hstvertex] (hA2) at (4, -13) {$A_2$};
\node[hstvertex] (hB2) at (4, -14) {$B_2$};
\node[hstvertex] (hD2) at (6, -13) {D};
\node[hstvertex] (hG2) at (6, -14) {G};
\node[hstvertex] (hH2) at (8, -13) {$H_2$};
\node[hstvertex] (hI2) at (7.5, -14) {I};
\node[hstvertex] (hJ2) at (8.5, -14) {J};
\node[hstvertex] (hA3) at (2, -15.5) {A};
\node[hstvertex] (hC3) at (2.75, -15.5) {C};
\node[hstvertex] (hB3) at (3.75, -15.5) {B};
\node[hstvertex] (hE3) at (4.5,  -15.5) {E};
\node[hstvertex] (hF3) at (5.25, -15.5) {F};
\node[hstvertex] (hD3) at (6.25, -15.5) {D};
\node[hstvertex] (hG3) at (7.25, -15.5) {G};
\node[hstvertex] (hH3) at (8.25, -15.5) {H};
\node[hstvertex] (hK3) at (9, -15.5) {K};
\node[hstvertex] (hI3) at (10, -15.5) {I};
\node[hstvertex] (hJ3) at (11, -15.5) {J};
\draw (hD1) -- (hA1) {};
\draw (hD1) -- (hH1) {};
\draw (hA2) -- (hB2) {};
\draw (hD2) -- (hG2) {};
\draw (hH2) -- (hI2) {};
\draw (hH2) -- (hJ2) {};

\node[ellipse, draw=red, minimum height = 90pt, minimum width = 25pt] at (6, -10.5) {};
\node[ellipse, draw=red, minimum height = 60pt, minimum width = 25pt] at (4, -13.5) {};
\node[ellipse, draw=red, minimum height = 60pt, minimum width = 25pt] at (6, -13.5) {};
\node[circle, draw=red, minimum width = 60pt] at (8, -13.5) {};
\node[ellipse, draw=red, minimum height=25pt, minimum width=45pt] at (2.35, -15.5) {};
\node[ellipse, draw=red, minimum height=30pt, minimum width=65pt] at (4.5, -15.5) {};
\node[circle, draw=red, minimum width = 20pt] at (6.25, -15.5) {};
\node[circle, draw=red, minimum width = 20pt] at (7.25, -15.5) {};
\node[ellipse, draw=red, minimum height=25pt, minimum width=45pt] at (8.6, -15.5) {};
\node[circle, draw=red, minimum width=20pt] at (10, -15.5) {};
\node[circle, draw=red, minimum width=20pt] at (11, -15.5) {};
\node at (2, -7.25) {Grove of Tree T:};
\node at (5.35, -9) {$T$};
\node at (3.35, -12.5) {$A_1$};
\node at (5.35, -12.5) {$D_1$};
\node at(7, -12.5) {$H_1$};
\node at (1.5, -15) {$A_2$};
\node at (3.25, -15) {$B_2$};
\node at (7.75, -15) {$H_2$};
\end{tikzpicture}
\end{center}

\caption{An example of a LDD, the corresponding HST, and the corresponding grove.}
\end{figure}

\section{Pricing Monotone Algorithms}
\label{sec:price}

In this section, we show that an algorithm for the online metrical matching can be implemented as a posted-price algorithm if and only if the algorithm satisfies the following monotonicity property. Intuitively, an algorithm is monotone if a request and server pair get matched with nondecreasing probability as the request is dragged towards the server. We note that monotonicity does not have a natural interpretation within the context of online metrical searching, which explains why we give a monotone algorithm for online metrical matching, even though we only analyze its competitiveness for online metrical search.

\begin{definition}
\label{definition:monotonicity}
    An algorithm $A$ for online metric matching is monotone if for every instance, every request
    $r_t$ in that instance, every possible
    sequence $R$ of random events internal to $A$ 
    prior to $r_t$'s arrival, and
    all vertices $u, v, s$ where $v$ is on the path from
    $u$ to $s$ it 
    is the case that: 
    $\Prob{A_R(r_t) = s \mid  E_R \text{ and }   r_t  = u}\leq \Prob{A_R(r_t) = s \mid  E_R \text{ and }   r_t  = v}$ 
    where $A_R(r_t) = s$ is the event that $A$ matches 
    $r_t$ to $s$, and $E_R$ is the event that the past random events internal to $A$ are equal to $R$. 
\end{definition}

\begin{theorem}
Any algorithm $A$ for the online metrical matching problem can be implemented as a posted-price algorithm if and only if $A$ is monotone.
\end{theorem}

In section \ref{sec:price-determ}, we show this correspondence for the deterministic setting and then extend it to randomized algorithms in section \ref{sec:rand-pricing}.
Lastly, in section \ref{nomononoprice} we show that all pricing schemes induce monotone matching algorithms, giving us the equivalency of pricing schemes and monotone matching algorithms. For notational simplicity we use  $d(u, v)$ instead of $d_T(u,v)$ in this section.

\subsection{Pricing Deterministic Monotone Algorithms}
\label{sec:price-determ}

%In this subsection, we prove in  Theorem   \ref{lem:det-mono-price} that all monotone deterministic algorithms can be implemented as pricing algorithms. 
We first need to  define monotone partitions. 

 \begin{definition} A monotone partition of the tree metric space
$T = (V, E, d_T)$ consists of 
two components. The first component is a partition $P = \{Q_1, \ldots, Q_t\}$ of the vertices $V$,
 such that
 for
 each  part $Q_i \in {P}$ it is the case that  the induced subgraph on $Q_i$
 is connected. The second component consists
  of a designated leader for each nonempty $Q_i \in P$, where a leader is an
 available server located in $Q_i$. 
\end{definition}

A monotone algorithm for serving a request can be simply derived from
a monotone partition by  matching a request in each part
to that part's designated leader. The converse is proved
in the course of the proof of lemma \ref{lem:det-mono-price}.
\begin{lemma}
  \label{lem:det-mono-price}
Every monotone  deterministic algorithm $\A$ for  online metrical matching  can be
implemented by a pricing algorithm $\B$. 
\end{lemma}

\begin{proof}
We first explain how to derive monotone partition $P$ from
$\A$ at each time step $t$. 
For each $s_i \in S_t$, let  part $Q_i$ consist of the vertices $v$
such that a request arriving on vertex $v$ would be served by $s_i$.  The leader of each nonempty 
$Q_i$ is $s_i$. The fact that each part $Q_i$ induces
a connected subgraph follows directly from the monotonicity 
of the algorithm. Let $u_{i, j}$ be the vertex in $Q_i$ closest to $s_j$.

 We now define the pricing scheme $p:S_j\rightarrow\R$ for time step $t$ for
 pricing algorithm $\B$.
 
\SetArgSty{textnormal}
  \begin{algorithm}[H]{
      \Indp 
      Set $p(s_i) = 0$ for   an arbitrary nonempty partition $Q_i \in P$ \\
      \For{every nonempty $Q_j$ whose leader is not yet priced and that is adjacent to a part $Q_i$ whose leader is already priced}{
        Set $p(s_j) = p(s_i) +d(u_{i,j}, s_i)-d(u_{j, i}, s_j)$
      }
      \For{any server $s_i\in S_t$ that is not already priced}{
        Set $p(s) = +\infty$
      }
      \Indm }
    %\caption{Pricing Function $p$\label{alg:p}}
  \end{algorithm}

  To show that that this pricing scheme implements $\A$, we will show that if $\A$ services a request at a vertex $v$ by server
  $s_i$ then  $d(v, s_i)+p(s_i) < d(v,
  s_j)+p(s_j)$ for all $j\neq i$. Assume otherwise to 
  reach a contradiction. 
   Let $s_j$ be the  server that  minimizes $d(v, s_j) + p(s_j)$ (with ties broken arbitrarily). We now break the proof into cases.
   
   In the first case assume that
 $Q_i$ and $Q_j$ are adjacent parts, or equivalently that the edge $(u_{i,j}, u_{j,i})$ connects $Q_i$ and $Q_j$. Then we can conclude that:
  \begin{align*}
      p(s_i) &= p(s_j) + d(u_{j, i}, s_j) - d(u_{i, j}, s_i) &\text{by definition of $p(s_i)$}\\
      &\le p(s_i) + d(v, s_i)-d(v, s_j) + d(u_{j, i}, s_j) - d(u_{i, j}, s_i) & \text{by assumption}\\
      &= p(s_i)+d(v, s_i) - d(v, u_{j, i}) - d(u_{i, j}, s_i) & \text{as $u_{j, i}$ is on the path from $v$ to $s_j$}\\
      &\le p(s_i) + d(v, u_{i, j}) - d(v, u_{j, i})&\text{by triangle inequality}\\
        &\leq p(s_i) - d(u_{i,j}, u_{j,i}) & \text{as $u_{i, j}$ is on the path from $v$ to $u_{j, i}$} \\
      &< p(s_i) & \text{by definition of metric space}
  \end{align*}
  Note that the first inequality holds independently of which of $Q_i$ and $Q_j$ was priced first, so our assumption causes a contradiction.

  In the second case assume that $Q_j$ and  $Q_i$ are not adjacent, i.e. that is there is
  not an edge of the form $(u_{i,j}, u_{j,i})$ in $T$. However, let $Q_k$ be the
  last nonempty part before $Q_j$ on the unique path from $s_i$ to $s_j$ in $T$. Thus the  edge $(u_{k, j},
  u_{j, k})$ exists in $T$. Then we can conclude that:
  \begin{align*}
      d(v, s_k)+p(s_k)&\leq d(v, u_{k, j}) + d(u_{k, j},
  s_k)+p(s_k) & \text{by triangle inequality}\\
  &= d(v, u_{k, j}) + d(u_{k, j},
  s_k)\\&\ \ \ +d(u_{j, k}, s_j) - d(u_{k, j}, s_k) + p(s_j) & \text{by definition of $p(s_k)$}\\
  &= d(v, u_{k, j})+d(u_{j, k}, s_j)+p(s_j)\\
  &= d(v, u_{k, j})+d(u_{k, j}, s_j) - d(u_{j,k}, u_{k,j})+p(s_j) & \text{as $u_{j, k}$ is on the path}\\[-.5em]& &\text{from $u_{k,j}$ to $s_j$}\\
    &< d(v, u_{k, j})+d(u_{k, j}, s_j)+p(s_j) & \text{by definition of metric space}\\
  &= d(v, s_j)+p(s_j)& \text{as $u_{k, j}$ is on the path}\\[-.5em]& &\text{from $v$ to $s_j$.}
  \end{align*} This again is a contradiction 
  to our minimality assumption for $s_j$. 
\end{proof}

\subsection{Randomized Algorithms}
\label{sec:rand-pricing}

From a randomized monotone algorithm $\mathcal{A}$ we derive
a corresponding distribution $\mathcal{P}$ over monotone partitions such that can implement $\mathcal{A}$ by picking a random
monotone partition $P$ from $\mathcal{P}$, and setting prices
as described in subsection \ref{sec:price-determ}. 
We will need the following definitions.

\begin{definition}
  ~

\begin{itemize}
    \item Let $\pi_i^v$  be the probability that algorithm $\mathcal{A}$ matches a request at $v$ to $s_i$.
    \item
     Let $\Pr_{\Pc}(P)$ denote the probability of partition $P$  under distribution
  $\Pc$. 

    \item
    Let $v\rightarrow_P s$ denote  that it is the case that in monotone partition ${P}$, 
    server $s$ is the leader of 
    the part $Q$ where $v\in Q\in{P}$.
    \item Let $v\rightarrow_{\Pc} s$ be the event that in a sampled $P\sim\Pc$ it holds that $v\rightarrow_P s$.
   \item Let $\Pc_i(w)$  be the monotone partitions in $\Pc$  that have $s_i$ as the leader of the part containing $w$.
\end{itemize}
\end{definition}

Our goal for the rest of this subsection is to prove the following lemma,
which asserts the existence of an appropriate $\Pc$. 
 
\begin{lemma} 
  \label{MainLemma} 
Consider a monotone matching algorithm $\mathcal{A}$ for the online metrical matching  on a tree $T$, a particular request $r_t$, and a particular
collection $S$ of remaining available servers. Then there
  exists a distribution $\Pc$ over monotone partitions of $T$ such that: $\Pr(w\rightarrow_\Pc s_i\ |\ r_t=w) = \pi_i^w$ for all
  servers $s_i \in S$ and all vertices $w$. 
\end{lemma}

\subsubsection{The Construction of $\Pc$}
  
  Without loss of generality it is sufficient to take $t=1$.
  The proof will be via induction on the number of vertices in $T$.
  The base case is the case where $T$ is a single vertex $v$ containing all $n$
  available servers on it. Then $\mathcal{P}$ consists of $n$
  monotone partitions, where monotone partition $P_i$ consists of one part
  containing vertex $v$ with server $s_i$ as the leader, and
  with associated probability $\pi^v_i$.

  For the inductive step, pick an
  arbitrary leaf $u$. Let $v$ be the unique neighbor of $u$ in $T$.  By renumbering assume servers $s_1, \ldots, s_m$ are
  located at $u$.  Let $T'$ be the tree derived from $T$ by deleting $u$, and moving servers $s_1, \ldots, s_m$ to $v$.
By induction there exists a probability distribution $\P'$ over monotone
  partitions of $T'$ such that
  \begin{equation}\label{inductiveeq}
      \sum_{P \in \Pc'_i(w)}\Pr_{\P'}(P) = \pi^w_i
  \end{equation} for all vertices $w$ in $T'$ and servers $s_i$.
 We now  obtain $\Pc$ from  $\Pc'$ 
 by extending each monotone partition $P$ in $\Pc'$ to a collection of monotone partitions
 in $\Pc$. So consider an arbitrary $P \in \Pc'$. 
We consider two cases. 

\paragraph{Case 1:} $P$ is of type 1 if $P\in\Pc'_i(v)$ for some $1\leq i\leq m$; That is, one of $s_1, \ldots , s_m$ is the leader of the part $Q$ satisfying
$v\in Q\in P$. There will be one
partition $P_1$ in $\Pc$ derived from $P$.
The partition $P_1$ is identical to $P$ except that the vertex $u$ is added to the part $Q$. $P_1$ inherits the probability of $P$, that is:
\begin{equation}\label{P1eq}
\Pr_\Pc(P_1) = \Pr_{\Pc'}(P).
\end{equation}

%\newpage
\paragraph{Case 2:} $P$ is of type 2 if $P\in\Pc'_i(v)$ for some $i>m$; That is, one of $s_{m+1}, \ldots , s_n$ is the leader of the part $Q$ satisfying
$v\in Q\in P$. There will be $m+1$
partitions $P_1, \ldots, P_{m+1}$ 
in $\Pc$ derived from $P$. 
For $j \in [1, m]$ the partition $P_j$
is identical to $P$ except that
$P_j$ contains a new part
consisting of only the vertex $u$
with the leader of this part being
$s_j$.
Partition $P_{m+1}$ is identical to $P$
except that in $P_{m+1}$ the part $Q$ satisfying $v \in Q  \in P$ also contains the vertex  $u$.

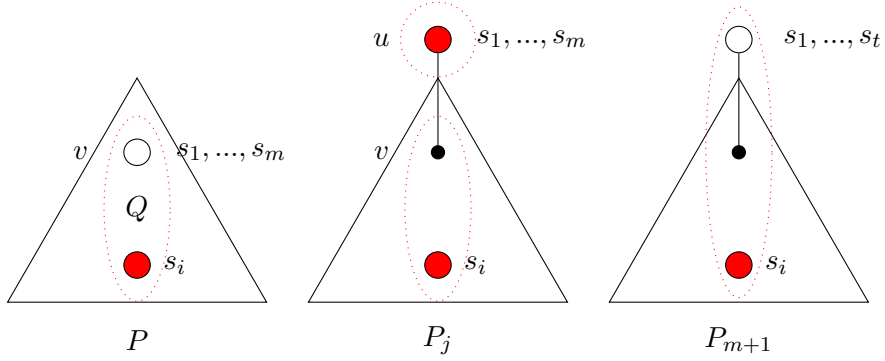
\begin{figure}[h]
\begin{center}\hspace{.3 in}
\begin{tikzpicture}
\tikzstyle{server}=[circle,fill=white, draw = black,minimum size=10pt,inner sep=0pt]
\tikzstyle{vertex}=[circle,fill=black, draw = black,minimum size=5pt,inner sep=0pt]
\tikzstyle{tree} = [regular polygon, regular polygon sides = 3, inner sep = 20 pt, draw = black]
\node[server, fill = red] (a) at (0, .5) {};
\node at (1.25, .5) {$s_1, ..., s_m$};
\node[tree] (c) at (0, -2) {};
\node[vertex] (b) at (0, -1) {};
\node[server, fill=red] (d) at (0, -2.5){};
\node[circle, draw= red, dotted, inner sep = 10 pt] at (0, .5) {};
\node[ellipse, minimum height = 70pt, minimum width = 25pt, draw = red, dotted] at (0, -1.75) {};
\node at (.5, -2.5) {$s_i$};
\draw (a) -- (b) {};
\node at (-4, -1.75) {$Q$};

\node[tree] (c') at (-4, -2) {};
\node[server] (a') at (-4, -1) {};
\node[server, fill=red] (d) at (-4, -2.5){};
\node at (-2.75, -1) {$s_1, ..., s_m$};
\node[ellipse, minimum height = 70pt, minimum width = 25pt, draw = red, dotted] at (-4, -1.75) {};
\node at (-4, -3.5) {$P$};
\node at (0, -3.5) {$P_j$};
\node at (4, -3.5) {$P_{m+1}$};
\node at (-3.5, -2.5) {$s_i$};

\node at (-4.75, -1) {$v$};
\node at (-.75, -1) {$v$};
\node at (-.75, .5) {$u$};

\node[server] (a) at (4, .5) {};
\node at (5.25, .5) {$s_1, ..., s_t$};
\node[tree] (c) at (4, -2) {};
\node[vertex] (b) at (4, -1) {};
\node[server, fill = red] (d) at (4, -2.5){};
\node[ellipse, minimum height = 110pt, minimum width = 25pt, draw = red, dotted] at (4, -1) {};
\node at (4.5, -2.5) {$s_i$};
\draw (a) -- (b) {};
\end{tikzpicture}
\end{center}
\caption{An example of case 2 $P_j$ for $1\leq j \leq m$ and $P_{m+1}$}
\end{figure}
%\newpage
To set the probabilities on $P_1, \ldots, P_{m+1}$ in $\Pc$ let $\delta_k = \pi^u_k-\pi^v_k$ for $k \in [1, m]$ and let $\Delta = \sum_{k=1}^m \delta_k$.
For each $j\in [1, m]$, we set 
\begin{equation}\label{Pjeq}
    \Pr_\Pc(P_j) =\Pr_{\Pc'}(P) \left(\frac{(\pi^u_j-\pi^v_j)(\pi^v_i-\pi^u_i)}{\Delta\cdot\pi^v_i}\right)
\end{equation} and we set
\begin{equation}\label{eqPm+1}
    \Pr_\Pc(P_{m+1}) =\Pr_{\Pc'}(P)\frac{\pi^u_i}{\pi^v_i}
\end{equation} 

The following two observations ensure that the
probability of $P_1, \ldots, P_{m+1}$ are
well-defined and that $\Pc$ is indeed a probability distribution.

\begin{observation}
For each $k\in\{1, ..., m\}$, we have that $\pi^u_k\geq\pi^v_k$; for each $k\in\{m+1, ..., n\}$, we have that $\pi^v_k\geq\pi^u_k$.
\end{observation}

\begin{proof}
This follows from the monotonicity of $\Ac$; it is assumed that $s_1, ..., s_m$ are all on vertex $u$, so for $k\in\{1, ..., m\}$ it must be the case that $\pi^u_k\geq\pi^v_k$. Since $u$ is a leaf vertex and $v$ is $u$'s only neighbor, $v$ must be on the path from $u$ to any other server. Thus $\pi^v_k\geq\pi^u_k$ for all $k\in\{m+1, ..., n\}$. 
\end{proof}

\begin{observation}
\label{obs:extension}
  $\sum_{j=1}^{m+1}\Pr_\Pc(P_j) = \Pr_{\Pc'}(P)$.
\end{observation}
\begin{proof}
\begin{align*}
    \sum_{j=1}^{m+1}\Pr_\Pc(P_j) &= \sum_{j=1}^m \Pr_{\Pc'}(P) \left(\frac{(\pi^u_j-\pi^v_j)(\pi^v_i-\pi^u_i)}{\Delta\cdot\pi^v_i}\right) + \Pr_{\Pc'}(P)\frac{\pi^u_i}{\pi^v_i} & \text{by Eq.s \ref{Pjeq} and \ref{eqPm+1}}\\
    &=\Pr_{\Pc'}(P)\frac{\pi^v_i-\pi^u_i}{\Delta\cdot\pi^v_i}\sum_{j=1}^m(\pi^u_j-\pi^v_j) + \Pr_{\Pc'}(P)\frac{\pi^u_i}{\pi^v_i}\\
    &=\Pr_{\Pc'}(P)\frac{\pi^v_i-\pi^u_i}{\Delta\cdot\pi^v_i}\Delta + \Pr_{\Pc'}(P)\frac{\pi^u_i}{\pi^v_i} & \text{by definition of $\Delta$}\\
    &=\frac{\Pr_{\Pc'}(P)}{\pi^v_i}\left(\pi^v_i-\pi^u_i+\pi^u_i\right)\\
    &=\frac{\Pr_{\Pc'}(P)}{\pi^v_i}\left(\pi^v_i\right) = \Pr_{\Pc'}(P).
\end{align*}
\end{proof}

\subsubsection{The analysis of $\Pc$}

We now turn to proving that our constructed $\Pc$
has the desired properties. 

\begin{definition}
  We partition the support of $\Pc$ as follows $$\text{supp}(\Pc) = \bigcup_{i=1}^m\Phi_i\cup\left(\bigcup_{i=m+1}^n\left(\Phi_i\cup\left(\bigcup_{j=1}^m\Pi_i^j\right)\right)\right)$$
where
\begin{equation}\label{eqForPhi}\Phi_i = \{P_1\ |\ P\in\Pc'_i(v)\}\end{equation} for each $1\leq i\leq m$,
\begin{equation}\label{eqforcapPi}
    \Pi_i^j = \{P_j\ | P\in\Pc'_i(v)\}
\end{equation}
for $i\in\{m+1, .., n\}$ and  $j\in\{1, ..., m\}$, and
\begin{equation}\label{Phieq2}
    \Phi_i = \{P_{m+1}\ |\ P\in\Pc'_i(v)\}
\end{equation}for $i\in\{m+1, ..., n\}$.
\end{definition}

\begin{lemma}
For $i\in\{1, ..., m\}$, we have that $\sum_{P\in\Phi_i}\Pr_\Pc(P) = \pi^v_i$.
\end{lemma}
\begin{proof}
\begin{align*}
    \sum_{P\in\Phi_i}\Pr_\Pc(P) &= \sum_{P\in\Pc'_i(v)}\Pr_{\Pc}(P_1) & \text{by Eq. \ref{eqForPhi}}\\
    &= \sum_{P\in\Pc'_i(v)}\Pr_{\Pc'}(P) & \text{by Eq. \ref{P1eq}}\\
    &=\pi^v_i & \text{by Eq. \ref{inductiveeq}}
\end{align*}
\end{proof}

\begin{lemma}
For $i\in\{m+1, ..., n\}$, we have that $\sum_{P\in\Phi_i}\Pr_\Pc(P) = \pi^u_i$.
\end{lemma}
\begin{proof}
\begin{align*}
    \sum_{P\in\Phi_i}\Pr_\Pc(P) &= \sum_{P\in\Pc'_i(v)}\Pr_{\Pc}(P_{m+1}) & \text{by Eq. \ref{Phieq2}}\\
    &= \sum_{P\in\Pc'_i(v)}\Pr_{\Pc'}(P)\frac{\pi^u_i}{\pi^v_i} & \text{by Eq. \ref{eqPm+1}}\\
    &= \frac{\pi^u_i}{\pi^v_i}\sum_{P\in\Pc'_i(v)}\Pr_{\Pc'}(P) & \text{by Eq. \ref{eqPm+1}}\\
    &=\frac{\pi^u_i}{\pi^v_i}\pi^v_i & \text{by Eq. \ref{inductiveeq}}\\
    &=\pi^u_i
\end{align*}
\end{proof}

\begin{lemma}   For $i\in\{m+1, ..., n\}$, we have that $\sum_{j=1}^m\sum_{P\in\Pi_i^j}\Pr_\Pc(P) = \pi^v_i-\pi^u_i$.
\begin{proof}
  
\begin{align*}
    \sum_{j=1}^m\sum_{P\in\Pi_i^j}\Pr_\Pc(P) &= \sum_{j=1}^m\sum_{P\in\Pc'_i(v)}\Pr_\Pc(P_{j}) & \text{by Eq. \ref{eqforcapPi}}\\
    &= \sum_{j=1}^m\sum_{P\in\Pc'_i(v)}\Pr_{\Pc'}(P) \left(\frac{(\pi^u_j-\pi^v_j)(\pi^v_i-\pi^u_i)}{\Delta\cdot\pi^v_i}\right) & \text{by Eq. \ref{Pjeq}}\\
    &= \sum_{j=1}^m\left(\frac{(\pi^u_j-\pi^v_j)(\pi^v_i-\pi^u_i)}{\Delta\cdot\pi^v_i}\right)\sum_{P\in\Pc'_i(v)}\Pr_{\Pc'}(P)\\
    &= \sum_{j=1}^m\left(\frac{(\pi^u_j-\pi^v_j)(\pi^v_i-\pi^u_i)}{\Delta\cdot\pi^v_i}\right)\pi^v_i & \text{by Eq. \ref{inductiveeq}}\\
    &= \frac{\pi^v_i-\pi^u_i}{\Delta}\sum_{j=1}^m(\pi^u_j-\pi^v_j)\\
    &= \frac{\pi^v_i-\pi^u_i}{\Delta}\Delta & \text{by definition of $\Delta$}\\
    &=\pi^v_i-\pi^u_i.
\end{align*}
\end{proof}
\end{lemma}

\begin{proof}[Proof of lemma \ref{MainLemma}]
By Observation \ref{obs:extension} we know that $\Pr(w\rightarrow_\Pc s_i\ |\ r_t=w) = \pi_i^w$  holds
for all vertices $w$ except for possibly for vertex $u$. Thus we just need to verify that this holds for $w=u$. 

For $i\in\{1, ..., m\}$, we have that 
\begingroup
\allowdisplaybreaks
\begin{align*}
    &\sum_{P\in\Pc_i(u)}\Pr_\Pc(P) = \sum_{P\in\Phi_i}\Pr_\Pc(P) + \sum_{j=m+1}^n\sum_{P\in\Pi_j^i}\Pr_\Pc(P)\\
    &= \sum_{P\in\Pc'_i(v)}\Pr_\Pc(P_1) + \sum_{j=m+1}^n\sum_{P\in\Pc'_j(v)}\Pr_\Pc(P_{i}) & \text{by Eq.s \ref{eqForPhi} and \ref{eqforcapPi}}\\
    &= \sum_{P\in\Pc'_i(v)}\Pr_{\Pc'}(P) + \sum_{j=m+1}^n\sum_{P\in\Pc'_j(v)}\Pr_{\Pc'}(P)\left(\frac{(\pi^u_i-\pi^v_i)(\pi^v_j-\pi^u_j)}{\Delta\cdot\pi^v_j}\right) & \text{by Eq.s \ref{P1eq} and \ref{Pjeq}}\\
    &= \pi^v_i + \sum_{j=m+1}^n\sum_{P\in\Pc'_j(v)}\Pr_{\Pc'}(P)\left(\frac{(\pi^u_i-\pi^v_i)(\pi^v_j-\pi^u_j)}{\Delta\cdot\pi^v_j}\right) & \text{by Eq. \ref{inductiveeq}}\\
    &= \pi^v_i + \sum_{j=m+1}^n\left(\frac{(\pi^u_i-\pi^v_i)(\pi^v_j-\pi^u_j)}{\Delta\cdot\pi^v_j}\right)\sum_{P\in\Pc'_j(v)}\Pr_{\Pc'}(P)\\
    &= \pi^v_i + \sum_{j=m+1}^n\left(\frac{(\pi^u_i-\pi^v_i)(\pi^v_j-\pi^u_j)}{\Delta\cdot\pi^v_j}\right)\pi^v_j & \text{by Eq. \ref{inductiveeq}}\\
    &= \pi^v_i + \frac{\pi^u_i-\pi^v_i}{\Delta}\sum_{j=m+1}^n(\pi^v_j-\pi^u_j)\\
    &= \pi^v_i + \frac{\pi^u_i-\pi^v_i}{\Delta}\sum_{j=1}^m(\pi^u_j-\pi^v_j) & \text{as both $\pi^u$ and $\pi^v$}\\[-1.5em]& &\text{are distributions}\\[-.5em]& &\text{over $1, ..., n$}\\
    &= \pi^v_i + \frac{\pi^u_i-\pi^v_i)}{\Delta}\Delta\\
    &= \pi^v_i+\pi^u_i-\pi^v_i = \pi^u_i
\end{align*}
\endgroup
For $i\in\{m+1, ..., n\}$, we have that 
\begingroup
\allowdisplaybreaks
\begin{align*}
    \sum_{P\in\Pc_i(u)}\Pr_\Pc(P) &= \sum_{P\in\Phi_i}\Pr_\Pc(P)\\
    &= \sum_{P\in\Pc'_i(v)}\Pr_{\Pc}(P_{m+1}) & \text{by Eq. \ref{Phieq2}}\\
    &=\sum_{P\in\Pc'_i(v)}\Pr_{\Pc'}(P)\frac{\pi^u_i}{\pi^v_i} & \text{by Eq. \ref{eqPm+1}}\\
    &=\frac{\pi^u_i}{\pi^v_i}\sum_{P\in\Pc'_i(v)}\Pr_{\Pc'}(P)\\
    &=\frac{\pi^u_i}{\pi^v_i}\pi^v_i=\pi^u_i & \text{by Eq. \ref{inductiveeq}}
\end{align*}
\endgroup This gives us that $\Pr(u\rightarrow_\Pc s_i) = \pi^u_i$ for all $i\in\{1, ..., n\}$.
\end{proof}

\subsection{Pricing Schemes Induce Monotone Matching Algorithms}\label{nomononoprice}

In this section, we show matching requests in an instance of  online metrical matching  according to a pricing scheme gives us a monotone algorithm. Specifically, given servers $S$ and a pricing scheme $\Pc$, let $p$ be the pricing function created by $\Pc$ prior to the arrival of the first request, and let $f_1: V\rightarrow S$ be the matching function such that $f(v) = \min_{s\in S}d(v, s) + p(s)$. Since any instance of  online metrical matching  after $j$ servers have been used is equivalent to an instance of  online metrical matching  with $n-j$ initial available servers, it's sufficient to just show that matching according to the pricing scheme is monotone on the first request. Now, let $s\in S$ and $u, v\in V$ such that $v$ is on the path from $u$ to $s$. Then, since $d(v, s)\leq d(u, s)$, we have that $\Pr(f(u)=s)\leq\Pr(f(v)=s)$. Thus $f$ is a monotone matching function. Since the associated matching function of any pricing scheme is monotone, this gives us the equivalence of pricing schemes and monotone matching algorithms for  online metrical matching .

\section{The Algorithm \TM}
\label{sect:search}

In subsection \ref{subsec:treematch} we define algorithm \TS\ for the metric search problem on a tree
$T = (V, E)$  rooted
at vertex a $\rho$. The distance metric on $T$ will not be of interest to us in this section. We will use the interpretation
of a car moving when its parking spot is decommissioned,
as introduced earlier, as we think that this interpretation
is more intuitive. 
The description of \TS\ in subsection \ref{subsec:treematch}
uses  a probability distribution $q^\sigma_t(\tau)$ that is complicated to define, so its exact definition is postponed until
subsection \ref{subsec:qdefine}, in which we also show that it achieves our goal of matching the experts distribution.
In subsection \ref{subsec:treematchcost} we analyze the number of jumps used by the \TS\ algorithm. 
Finally in subsection \ref{subsec:treematchmonotonicity}, we show how to
convert \TS\ into a monotone  algorithm \TM\ for online
metrical matching that is identical to \TS\ on online metrical search instances.

%- note that since all decommissionings are synonymous with specific locations in the metric space, we will often refer to the location of a decommissioning $r_t$ by $r_t$ itself when convenient. If $v$ is a descendant of a leaf-spot $w_i$, then instead set $L(v) = \{w_i\}$. \todo{Are we sure this last sentence is what want?} \todo{Yes, it's more of a hack than anything but it works}

\begin{comment}
\todo{This is part of the algorithm description}
Our algorithm begins by continually moving the car towards the root until there are no more parking spots available. Upon this moment, we instantiate an instance of the standard multiplicative weights online learning algorithm where we conceptually associate each expert with an $\sigma\in\{1, ..., d\}$ which recommends keeping $r_1$ on $T_\sigma$;
\TS\ parks $c$ on a random subtree with a probability distribution
carefully defined so that the probability that $c$ is on $T_\sigma$
%(with the right number of holes) 
is equal to the probability
that the multiplicative weights algorithm gives to this expert. 
\end{comment}

\subsection{Algorithm Description}
\label{subsec:treematch}

We start with some needed definitions and notation.

\begin{definition}
A parking spot $s_i$ in the collection  $ S$ of parking spots is a leaf-spot if there are no other parking spots in the
subtree rooted at $s_i$. 
Let
$L(T) = \{\ell_1, ..., \ell_d\}$ denote the collection of leaf-spots. 
Let $H$ be the maximum initial number of parking spots in $T$ on the path from the root $\rho$ to a leaf-spot in $L(T)$. 
For $\sigma\in  [d]$, define $T_\sigma\subseteq V$ as the set of parking spots on the path from the root $\rho$ to $\ell_\sigma$, inclusive. 
We define $T_\sigma$ to be alive if there is still an
in-commission parking spot in $T_\sigma$, and dead otherwise.
A $T_\sigma$ is killed by $r_t$ if $r_t$ 
 is the last parking spot to be decommissioned in
 $T_\sigma$. 
 Let $\Ac_t = \{\sigma\in[d] \ |\ \text{$T_\sigma$ is alive just before the arrival of $r_t$}\}$. 
For a vertex $v \in V$, let $L(v)$ denote the collection of leaf-spots that are descendants of $v$ in $T$.
Let $c_t$ be the location of the car just before the arrival
of request $r_t$.
\end{definition}

\noindent
{\bf Algorithm \TS :} The algorithm has two phases: the prologue phase and the core phase. The algorithm starts in the prologue phase and  transitions to the core phase after the first time $m$ when there is
no available parking space on the path from the new parking spot $c_{m+1}$ to the root $\rho$, inclusive. 
The algorithm then remains in the core phase until the end. In the prologue phase, 
whenever the car is not parked at a vertex with an in-commission parking spot, the following actions are taken:
\begin{enumerate}
    \item If there is an in-commission parking spot at   $c_t  $ then no action is taken. 
    \item Else if there is an in-commission parking spot on the path between $c_t$  and the root $\rho$, inclusive, then the car moves to the first in-commission parking spot on this path nearest to $c_t$.  
    \item 
    Else the car moves to the root $\rho$ and enters the core phase to determine where to go from there. So for analysis purposes, the movement to $\rho$ counts as being part of the prologue phase, and the rest of the movement counts as being in the core phase.  
\end{enumerate}
If the car is at the root $\rho$ and the algorithm is just transitioning into the core phase, then 
   a live $T_\tau$ is picked uniformly at random
    from $\Ac_{t+1}$, 
    an internal variable $\gamma$ is set to be $ \tau $, and the car moves to the first
    in-commission parking spot on the path from $\rho$ to $\ell_\tau$. 
Subsequently in the core phase, 
when a parking spot $r_t$ is decommissioned then:
\begin{enumerate}
    \item If the car is not parked at $r_t$, that is
    if $c_t \ne r_t$, then no action is taken. 
    \item
    Else the car moves to the first in-commission parking spot in $T_\tau$ with probability $q^{\gamma}_t(\tau)$ and sets 
    $\gamma$ to be  $ \tau$.  ( $q^{\gamma}_t(\tau)$ is defined in the next subsection.) 
\end{enumerate}
Intuitively $\gamma$ stores the last random choice of the algorithm.

\subsection{The Definition of $q^\sigma_t(\tau)$}
\label{subsec:qdefine}

In this section we only consider times in the core phase. 
We conceptually divide up the tree $T$ into three regions. Given vertex $v$ and time $t$, we let $z^v_t$ be the number of in-commission parking spots on the path from $v$ to $\rho$, inclusive, just before decommission $r_t$. We then define the regions as follows: 
\begin{enumerate}
    \item The \emph{root region} is the set of all vertices $v$ such that $z^v_t = 0$. Note that this region is connected, and no decommissioning can occur in this region since there are no parking spots left. 
    \item The \emph{frontier region} is the set of all vertices $v$ such that $z^v_t = 1$. A decommissioning $r_t$ is called a frontier decommissioning if $r_t$ is in the frontier region.
    \item The \emph{outer region} is the set of all vertices $v$ such that $z^v_t > 1$. A decommissioning $r_t$ is called a outer decommissioning if $r_t$ is in the outer region.
\end{enumerate} 
Observe that these regions  have no dependence on random events internal to the algorithm. Further observe that step 2 of the core phase in algorithm \TS\ maintains the invariant that the car is always parked at a spot in the frontier region. This means that any outer decommissionings will not move the car from its current parking spot.

\begin{definition}
Let $r_m$ be the last decommissioning handled in the prologue phase of \TS. Define $\Xc_t = \Ac_t\cap L(r_t)$
to be the collection of $\sigma$'s such that 
$T_{\sigma}$ is alive and contains  $r_t$ and define $\Yc_t = \Ac_t\setminus\Xc_t = \Ac_t\cap\overline{L(r_t)}$
to be the collection of $\sigma$'s such that 
$T_{\sigma}$ is alive and doesn't contain  $r_t$. Define $\mathcal{F}_t = \Xc_t \cap \overline{\Ac_{t+1}}$ to be the collection of $\sigma$'s such that $T_\sigma$ is killed by $r_t$.
Let $n_t^\sigma$ denote the number of frontier decommissionings strictly before time $t$ from $T_\sigma$.
Define $w^\sigma_t = (1-\epsilon)^{n^\sigma_t}$ for each $\sigma\in [d]$. 
Define $W_t(\mathcal{J}) = \sum_{\sigma\in \mathcal{J}} w^\sigma_t$ for any $\mathcal{J}\subseteq\{1, ..., d\}$. Define $\pi^\sigma_t$ as the probability the experts algorithm would give to expert $\sigma$, that is $\pi^\sigma_t = \frac{w^\sigma_t}{\sum_{\tau\in[d]}w^\tau_t}$. Define $\tilde{\pi}^\sigma_t$ as $\pi_t$ normalized amongst all experts in $\Ac_t$, that is $\tilde{\pi}^\sigma_t =  \frac{w^\sigma_t}{\sum_{\tau\in\Ac_t}w^\tau_t}$ if $\sigma\in\Ac_t$, and 0 otherwise. Define $p^\sigma_t$ as the probability that
$\gamma = \sigma$ right before time $t$. 
\end{definition}
We are now ready to define $q_t^\sigma(\tau)$. Note that by the definition of \TS, $q_t^\sigma(\tau)$ is only used for $\sigma \in \Xc_t$ since the algorithm only reaches step 2 of the core phase when $r_t \in T_\gamma$. 
We  show in Lemma \ref{lem:qsetting} that this definition of $q^{\sigma}_t(\tau)$ indeed defines a probability distribution over $\tau\in [d]$. 
We then show in Lemma \ref{lem:tminvariant}  that
the definition of $q^{\sigma}_t(\tau)$ guarantees that our desired invariant $p^\sigma_t = \tilde{\pi}^\sigma_t$ holds. 

\begin{definition}
\label{defn: qsetting}
\begin{equation*}%\label{eq: qsetting}
q^{\sigma}_t(\tau) = \begin{cases}
 \frac{\epsilon w^\tau_t}{(1-\epsilon)W_t(\Xc_t \setminus \mathcal{F}_t) + W_t(\Yc_t)} & \text{if } \tau\in\Yc_t\text{ and $\sigma\in \Xc_t \setminus \mathcal{F}_t$}\\[1em]
 \frac{w^\tau_t}{(1-\epsilon)W_t(\Xc_t \setminus \mathcal{F}_t) + W_t(\Yc_t)} & \text{if } \tau\in\Yc_t\text{ and $\sigma\in \mathcal{F}_t$}\\ \\
 \frac{1-\sum_{\ss \in \Yc_t}q^{\sigma}_t(\ss)}{|\Xc_t\setminus{\mathcal{F}_t}|} & \text{if } \tau\in\Xc_t\setminus \mathcal{F}_t\\
 0 & \text{if } \tau\in \mathcal{F}_t \text{ or } \tau \in \overline{\Ac_t}
\end{cases}\end{equation*}
\end{definition}

\begin{lemma}\label{lem:qsetting}
For all times $t$ in the core phase and for all $\sigma \in \Xc_t$,
 $q^\sigma_t(\tau)$ forms a distribution over $\tau\in [d]$.
 \end{lemma}
\begin{proof} 
First, note that the cases of Definition $\ref{defn: qsetting}$ partition $[d]$ since 
\begin{align*}
    [d] &= \Ac_t \cup \overline{\Ac_t}  & \text{by definition of $\Ac_t$}\\
    &= (\Xc_t \cup \Yc_t) \cup \overline{\Ac_t} & \text{by definition of $\Xc_t$, $\Yc_t$}\\
    &= ((\Xc_t \setminus \mathcal{F}_t) \cup \mathcal{F}_t \cup \Yc_t) \cup \overline{\Ac_t} & \text{by definition of $\mathcal{F}_t$}
\end{align*}
Next, we show that the distribution sums to 1.
\begin{align*}
    \sum_{\tau \in [d]} q^\sigma_t(\tau) &= \sum_{\tau \in \Xc_t \setminus \mathcal{F}_t}q^\sigma_t(\tau) + \sum_{\tau \in \Yc_t}q^\sigma_t(\tau) & \text{since $q^\sigma_t(\tau) = 0$ for $\tau \in \mathcal{F}_t$ or $\tau \in \overline{\Ac_t}$} \\
    &= \sum_{\tau \in \Xc_t \setminus \mathcal{F}_t} \frac{1-\sum_{\ss \in \Yc_t}q^{\sigma}_t(\ss)}{|\Xc_t\setminus{\mathcal{F}_t}|} + \sum_{\tau \in \Yc_t}q_t^\sigma(\tau) & \text{By Defn. 15} \\
    &= |\Xc_t\setminus{\mathcal{F}_t}|\frac{1-\sum_{\ss \in \Yc_t}q^{\sigma}_t(\ss)}{|\Xc_t\setminus{\mathcal{F}_t}|} + \sum_{\tau \in \Yc_t}q_t^\sigma(\tau) \\
    &= 1
\end{align*}
Finally, we show that every event has non-negative probability.  Trivially, $q_t^\sigma(\tau) \geq 0$ for any $\tau \in [d] \setminus (\Xc_t \setminus \mathcal{F}_t)$ by Definition $\ref{defn: qsetting}$.  For $\tau \in \Xc_t \setminus \mathcal{F}_t$, $q_t^\sigma(\tau) = \frac{1-\sum_{\ss \in \Yc_t}q^{\sigma}_t(\ss)}{|\Xc_t\setminus{\mathcal{F}_t}|}$, so it is sufficient to show that $\sum_{\tau \in \Yc_t}q^{\sigma}_t(\tau) \leq 1$.  If $\sigma \in X_t \setminus \mathcal{F}_t$, then
\begin{align*}
    \sum_{\tau \in \Yc_t}q^{\sigma}_t(\tau) &= \sum_{\tau \in \Yc_t}\frac{\epsilon w^\tau_t}{(1-\epsilon)W_t(\Xc_t \setminus \mathcal{F}_t) + W_t(\Yc_t)} & \text{by Defn. $\ref{defn: qsetting}$} \\
    &\leq \frac{\epsilon}{W_t(\Yc_t)}   \sum_{\tau \in \Yc_t} w^\tau_t & \text{ since $(1-\epsilon)W_t(\Xc_t \setminus \mathcal{F}_t) \geq 0$} \\
    &= \frac{\epsilon W_t(\Yc_t)}{W_t(\Yc_t)} & \text{by definition of $\Yc_t$} \\
    &= \epsilon \\ &\leq 1
\end{align*}

If $\sigma \in \mathcal{F}_t$, the only difference from the previous case is that the constant $\epsilon$ is replaced with a 1 in Definition $\ref{defn: qsetting}$, so the analysis still holds.

\end{proof}
Before proving that the invariant $p^\sigma_t = \tilde{\pi}_t^\sigma$ holds using distribution $q$, we prove some necessary properties about the set $X_t$.
\begin{lemma}\label{lem:xtsubset}
    Let $X \subseteq \Xc_t$ be an arbitrary subset of $\Xc_t$ and let $\sigma \in \Xc_t$ be any parking spot, we then have that:
\begin{enumerate}
    \item  $w_t^\sigma = \frac{W_t(X)}{|X|}$
    \item $\sum\limits_{\tau \in X}\tilde{\pi}_t^\tau = \frac{W_t(X)}{W_t(\Ac_t)}$
\end{enumerate}
\end{lemma}
\begin{proof}
  Given decommissioning $r_t$, let $\sigma_i, \sigma_j \in \Xc_t$ be arbitrary spots. Then, by definition of $\Xc_t$, $T_{\sigma_i}$ and $T_{\sigma_j}$ overlap at $r_t$ as well as every parking spot from $r_t$ to the root.  So, any frontier decommissioning before $r_t$ that decommissioned a spot from $T_{\sigma_i}$ also decommissioned a spot from $T_{\sigma_j}$.  In other words, $n_t^{\sigma_i} = n_t^{\sigma_j}$ and $w_t^{\sigma_i} = w_t^{\sigma_j}$.  Since $\sigma_i$ and $\sigma_j$ were arbitrary from $\Xc_t$, for any subset $X \subseteq \Xc_t$ and $\sigma \in \Xc_t$,
  \begin{align*}
     W_t(X) &= \sum_{\tau \in X} w_t^\tau = |X|w_t^\sigma 
  \end{align*}
  Property 2 comes just from the definition of $\tilde{\pi}_t^\sigma$:
  \begin{align*}
      \sum_{\tau \in X} \tilde{\pi}_t^\tau &= \frac{\sum_{\tau \in X} w_t^\tau}{\sum_{\tau \in \Ac_t} w_t^\tau} & \text{by defintion of $\tilde{\pi}_t^\tau$} \\
      &= \frac{W_t(X)}{W_t(\Ac_t)} & \text{by definition of $W_t(X)$, $W_t(\Ac_t)$}
  \end{align*}
\end{proof}

\begin{lemma}\label{lem:tminvariant}
For all times $t$ during the core phase and
for all $\sigma \in [d]$,   $p^\sigma_t = \tilde{\pi}^\sigma_t$.
\end{lemma}
\begin{proof}
The proof is  by induction on the time $t$. 
The transition from the prologue to core phases 
of \TS\ guarantees that this invariant is initially true.
Now assuming  $p^\sigma_t = \tilde{\pi}^\sigma_t$ for all $\sigma$, we want to show that $p^\sigma_{t+1} = \tilde{\pi}^\sigma_{t+1}$ for all $\sigma$. 
The proof is broken into cases.

In the first case, assume $\sigma\in\Xc_t \setminus \mathcal{F}_t$ and $r_t$ is a  frontier decommissioning.
We will begin with two equations:
\begin{equation}\label{eq: xsum1}
\begin{aligned}
  \sum\limits_{\tau \in \Xc_t \setminus \mathcal{F_T}}&p^\tau_t\frac{1-\sum_{\ss \in \Yc_t}q^{\tau}_t(\ss)}{|\Xc_t\setminus{\mathcal{F}_t}|} \\&=  \sum\limits_{\tau \in \Xc_t \setminus \mathcal{F_T}}p^\tau_t\frac{1-\sum_{\ss \in \Yc_t}\frac{\epsilon w^\ss_t}{(1-\epsilon)W_t(\Xc_t \setminus \mathcal{F}_t) + W_t(\Yc_t)}}{|\Xc_t\setminus{\mathcal{F}_t}|} & \text{by Defn. $\ref{defn: qsetting}$} \\
  &= \sum\limits_{\tau \in \Xc_t \setminus \mathcal{F_T}}p^\tau_t\frac{1-\frac{\epsilon W_t(\Yc_t)}{(1-\epsilon)W_t(\Xc_t \setminus \mathcal{F}_t) + W_t(\Yc_t)}}{|\Xc_t\setminus{\mathcal{F}_t}|} & \text{by defn. of $W_t(\Yc_t)$} \\
  &= \sum\limits_{\tau \in \Xc_t \setminus \mathcal{F_T}}p^\tau_t\frac{(1-\epsilon)W_t(\Ac_t \setminus \mathcal{F}_t)}{|\Xc_t\setminus{\mathcal{F}_t}|((1-\epsilon)W_t(\Xc_t \setminus \mathcal{F}_t) + W_t(\Yc_t))} & \text{by defn. of $\Ac_t, \mathcal{F}_t$} \\
  &= \frac{(1-\epsilon)W_t(\Ac_t \setminus \mathcal{F}_t)}{|\Xc_t\setminus{\mathcal{F}_t}|((1-\epsilon)W_t(\Xc_t \setminus \mathcal{F}_t) + W_t(\Yc_t))}\sum\limits_{\tau \in \Xc_t \setminus \mathcal{F_T}}\tilde{\pi}^\tau_t & p_t^\tau=\tilde{\pi}_t^\tau\\
  &=  \frac{(1-\epsilon)W_t(\Xc_t \setminus \mathcal{F}_t)W_t(\Ac_t \setminus \mathcal{F}_t)}{|\Xc_t\setminus{\mathcal{F}_t}|W_t(\Ac_t)((1-\epsilon)W_t(\Xc_t \setminus \mathcal{F}_t) + W_t(\Yc_t))} & \text{by Lemma $\ref{lem:xtsubset}$ part 2}
\end{aligned}
\end{equation}

and similarly,
\begin{equation}\label{eq: xsum2}
    \begin{aligned}
      \sum\limits_{\tau \in \mathcal{F_T}}&p^\tau_t\frac{1-\sum_{\ss \in \Yc_t}q^{\tau}_t(\ss)}{|\Xc_t\setminus{\mathcal{F}_t}|} \\&=  \sum\limits_{\tau \in \mathcal{F_T}}p^\tau_t\frac{1-\sum_{\ss \in \Yc_t}\frac{w^\ss_t}{(1-\epsilon)W_t(\Xc_t \setminus \mathcal{F}_t) + W_t(\Yc_t)}}{|\Xc_t\setminus{\mathcal{F}_t}|} & \text{by Defn. $\ref{defn: qsetting}$} \\
  &= \sum\limits_{\tau \in \mathcal{F_T}}p^\tau_t\frac{1-\frac{W_t(\Yc_t)}{(1-\epsilon)W_t(\Xc_t \setminus \mathcal{F}_t) + W_t(\Yc_t)}}{|\Xc_t\setminus{\mathcal{F}_t}|} & \text{by defn. of $W_t(\Yc_t)$} \\
  &= \sum\limits_{\tau \in \mathcal{F_T}}p^\tau_t\frac{(1-\epsilon)W_t(\Xc_t \setminus \mathcal{F}_t)}{|\Xc_t\setminus{\mathcal{F}_t}|((1-\epsilon)W_t(\Xc_t \setminus \mathcal{F}_t) + W_t(\Yc_t))} \\
  &= \frac{(1-\epsilon)W_t(\Xc_t \setminus \mathcal{F}_t)}{|\Xc_t\setminus{\mathcal{F}_t}|((1-\epsilon)W_t(\Xc_t \setminus \mathcal{F}_t) + W_t(\Yc_t))}\sum\limits_{\tau \in \mathcal{F_T}}\tilde{\pi}^\tau_t & p_t^\tau=\tilde{\pi}_t^\tau\\
  &=  \frac{(1-\epsilon)W_t(\Xc_t \setminus \mathcal{F}_t)W_t(\mathcal{F}_t)}{|\Xc_t\setminus{\mathcal{F}_t}|W_t(\Ac_t)((1-\epsilon)W_t(\Xc_t \setminus \mathcal{F}_t) + W_t(\Yc_t))} & \text{by Lemma $\ref{lem:xtsubset}$ part 2}
    \end{aligned}
\end{equation}
Now, notice that by the definition of $q_t^\sigma(\tau)$, 
if $\gamma = \sigma \in\Xc_t$ after $r_t$, then it must have been the case that $\gamma \in \Xc_t$ before $r_t$, and thus:
\begin{align*}p^\sigma_{t+1} &=\sum\limits_{\tau \in \Xc_t}{p}_t^\tau q_t^\tau(\sigma) \\
&= \sum\limits_{\tau \in \Xc_t}p^\tau_t\frac{1-\sum_{\ss \in \Yc_t}q^{\tau}_t(\ss)}{|\Xc_t\setminus{\mathcal{F}_t}|} & \text{by Defn. $\ref{defn: qsetting}$} \\
&= \sum\limits_{\tau \in \Xc_t \setminus \mathcal{F_T}}p^\tau_t\frac{1-\sum_{\ss \in \Yc_t}q^{\tau}_t(\ss)}{|\Xc_t\setminus{\mathcal{F}_t}|} + \sum\limits_{\tau \in \mathcal{F}_t}p^\tau_t\frac{1-\sum_{\ss \in \Yc_t}q^{\tau}_t(\ss)}{|\Xc_t\setminus{\mathcal{F}_t}|} \\
&= \frac{(1-\epsilon)W_t(\Xc_t \setminus \mathcal{F}_t)\left(W_t(\Ac_t \setminus \mathcal{F}_t) + W_t(\mathcal{F}_t)\right)}{|\Xc_t\setminus{\mathcal{F}_t}|W_t(\Ac_t)((1-\epsilon)W_t(\Xc_t \setminus \mathcal{F}_t) + W_t(\Yc_t))} & \text{by Eq. $\ref{eq: xsum1}$ and $\ref{eq: xsum2}$}\\
&= \frac{(1-\epsilon) W_t(\Xc_t \setminus \mathcal{F}_t)}{|\Xc_t\setminus{\mathcal{F}_t}|((1-\epsilon)W_t(\Xc_t \setminus \mathcal{F}_t) + W_t(\Yc_t))} & \text{by definition of $W(\Ac_t)$} \\
&=  \frac{(1-\epsilon) w_t^\sigma}{(1-\epsilon)W_t(\Xc_t \setminus \mathcal{F}_t) + W_t(\Yc_t)} & \text{by Lemma $\ref{lem:xtsubset}$ part 1} \\
&= \tilde{\pi}_{t+1}^\sigma & \text{by definition of $\tilde{\pi}_{t+1}^\sigma$}
\end{align*}

In the second case, assume that  $\sigma \in \Yc_t$ and $r_t$ is a  frontier decommissioning. We will begin with three equations:
\begin{equation}\label{pieq2}
    \begin{aligned}
\tilde{\pi}^\sigma_{t+1}-\tilde{\pi}^\sigma_t &= \frac{w^\sigma_t}{(1-\epsilon)W_t(\Xc_t \setminus \mathcal{F}_t) + W_t(\Yc_t)} - \frac{w^\sigma_t}{W_t(\Xc_t) + W_t(\Yc_t)} & \text{by defn. of } \pi_{t}^\sigma\\
&= \frac{(W_t(\mathcal{F}_t) + \epsilon\ W_t(\Xc_t \setminus \mathcal{F}_t))w^\sigma_t}{(W_t(\Xc_t)+W_t(\Yc_t))((1-\epsilon)W_t(\Xc_t) + W_t(\Yc_t))}\\
&= {\tilde{\pi}}_t^\sigma\frac{W_t(\mathcal{F}_t) + \epsilon\ W_t(\Xc_t \setminus \mathcal{F}_t)}{(1-\epsilon)W_t(\Xc_t) + W_t(\Yc_t)} & \text{by defn. of $\tilde{\pi}_t^\sigma$}\\
&= {p}_t^\sigma\frac{W_t(\mathcal{F}_t) + \epsilon\ W_t(\Xc_t \setminus \mathcal{F}_t)}{(1-\epsilon)W_t(\Xc_t) + W_t(\Yc_t)} & \text{as $p_t^\sigma = \tilde{\pi}^\sigma_t$}.
\end{aligned}
\end{equation}
\begin{equation}\label{eq: ysum1}
\begin{aligned}
\sum_{\tau\in\Xc_t\setminus\mathcal{F}_t}\tilde{\pi}_t^\tau q_t^\tau(\sigma) &= \sum_{\tau\in\Xc_t\setminus\mathcal{F}_t}\tilde{\pi}_t^\tau\frac{\epsilon w^\sigma_t}{(1-\epsilon)W_t(\Xc_t \setminus \mathcal{F}_t) + W_t(\Yc_t)} & \text{by Defn. $\ref{defn: qsetting}$} \\
&= \frac{\epsilon w^\sigma_t}{(1-\epsilon)W_t(\Xc_t \setminus \mathcal{F}_t) + W_t(\Yc_t)}\sum_{\tau\in\Xc_t\setminus\mathcal{F}_t}\tilde{\pi}_t^\tau \\
&= \frac{\epsilon w^\sigma_t}{(1-\epsilon)W_t(\Xc_t \setminus \mathcal{F}_t) + W_t(\Yc_t)} \left(\frac{W_t(\Xc_t \setminus \mathcal{F}_t)}{W_t(\Ac_t)}\right) & \text{by Lemma $\ref{lem:xtsubset}$ part 2} \\
&= \tilde{\pi}_t^\sigma \frac{\epsilon W_t(\Xc_t \setminus \mathcal{F}_t)}{(1-\epsilon)W_t(\Xc_t \setminus \mathcal{F}_t) + W_t(\Yc_t)} & \text{by defn. of $\tilde{\pi}_t^\sigma$} \\
&= p_t^\sigma  \frac{\epsilon W_t(\Xc_t \setminus \mathcal{F}_t)}{(1-\epsilon)W_t(\Xc_t \setminus \mathcal{F}_t) + W_t(\Yc_t)} & \text{as $p_t^\sigma = \tilde{\pi}_t^\sigma$}
\end{aligned}
\end{equation}
and similarly,
\begin{equation}\label{eq: ysum2}
\begin{aligned}
\sum_{\tau\in \mathcal{F}_t}\tilde{\pi}_t^\tau q_t^\tau(\sigma) &= \sum_{\tau\in\mathcal{F}_t}\tilde{\pi}_t^\tau\frac{ w^\sigma_t}{(1-\epsilon)W_t(\Xc_t \setminus \mathcal{F}_t) + W_t(\Yc_t)} & \text{by Defn. $\ref{defn: qsetting}$} \\
&= p_t^\sigma  \frac{W_t( \mathcal{F}_t)}{(1-\epsilon)W_t(\Xc_t \setminus\mathcal{F}_t) + W_t(\Yc_t)}
\end{aligned}
\end{equation}

If right after $r_t$ it was the case that $\gamma  = \sigma \in \Yc_t$, then it must
be the case that right before $r_t$
either $\gamma = \sigma$  and the car did not move
at time $t$, or the car moved and $\gamma$ was updated. This gives us
$$p^\sigma_{t+1} =p^\sigma_t+ \sum_{\tau\in\Xc_t}{p}_t^{\tau}q^\tau_t(\sigma)$$ or equivalently:
\begin{align*}
    &p^\sigma_{t+1}-p^\sigma_t = \sum_{\tau\in\Xc_t}{p}_t^{\tau}q^\tau_t(\sigma) \\
    &=\sum_{\tau\in\Xc_t}\tilde{\pi}^\tau_t q^\tau_t(\sigma) & p_t^\tau = \tilde{\pi}_t^\tau\\
    &= \sum_{\tau\in\Xc_t\setminus\mathcal{F}_t}\tilde{\pi}^\tau_t q^\tau_t(\sigma) + \sum_{\tau\in \mathcal{F}_t}\tilde{\pi}^\tau_t q^\tau_t(\sigma)\\
    &= p_t^\sigma  \frac{(\epsilon W_t(\Xc_t \setminus \mathcal{F}_t) + W_t(\mathcal{F}_t))}{(1-\epsilon)W_t(\Xc_t \setminus \mathcal{F}_t) + W_t(\Yc_t)} & \text{by Eq. $\ref{eq: ysum1}$ and $\ref{eq: ysum2}$} \\
    &= \tilde{\pi}^\sigma_{t+1}-\tilde{\pi}^\sigma_t & \text{by Eq. $\ref{pieq2}$}
\end{align*}
Hence we can conclude that 
$p_{t+1}^\sigma = \tilde{\pi}_{t+1}^\sigma$.

In the third case, assume $\sigma \in \mathcal{F}_t$ and $r_t$ is a frontier decommissioning.  By definition of $\mathcal{F}_t$, $T_\sigma$ has been killed by $r_t$ and so $p_{t+1}^\sigma = 0$.  Similarly, since $\sigma \notin \mathcal{A}_{t+1}$, $\tilde{\pi}_{t+1}^\sigma = 0 = p_{t+1}^\sigma$.

In the last case assume that  $r_t$ is an outer decommissioning. 
Then $r_t$ does not increase $n^\sigma_t$ for any $\sigma\in[d]$, so $\tilde{\pi}^\sigma_t = \tilde{\pi}^\sigma_{t+1}$. Moreover, by the invariant that we also keep the car at a parking spot in the frontier, $r_t$  cannot decommission the parking spot at $c_t$, and thus the car does not move. Thus 
$p^\sigma_{t+1} = p^\sigma_t $.
\end{proof}

\subsection{Cost Analysis}
\label{subsec:treematchcost}
Definition \ref{defn: qsetting} with Lemmas \ref{lem:qsetting} and \ref{lem:tminvariant} give us the following bound on the cost:

\begin{theorem}\label{thm:MSUnit}
During the prologue phase, 
$\sum_{t=1}^m \mathbf{1}^{\mathsf{TM}}(t)\leq H$
and during the core phase, 
$$\textbf{E}\left[\sum_{t=m+1}^{k-1}\mathbf{1}^{\mathsf{TM}}(t)\right]\leq (1+\epsilon)H + \frac{\ln{d}}{\epsilon}$$
where $\mathbf{1}^{\mathsf{TM}}(t)$ is an indicator random variable that is 1 if \TS\ moves the car to  a new parking spot on the decommissioning $r_t$ and $0$ otherwise.\end{theorem}\begin{proof} 
The first inequality for the prologue phase follows from the fact that the
car always moves towards the root and the definition of $H$. So now consider a time $t$ in the core phase,  and $\sigma\in[d]$. Let $\Dc_t = [d]\setminus\Ac_t$ be the set of dead paths just before decommissioning $r_t$ and define $$c^\sigma_t = \begin{cases}
1 & \sigma\in\Xc_t\text{ or }\sigma\in\Dc_t\\
0 & \text{otherwise.}
\end{cases}$$ 
Let $\delta_t^\sigma$ be defined for all $\sigma \in \mathcal{A}_t$ such that $\tilde{\pi}_t^\sigma = \pi_t^\sigma + \delta_t^\sigma$. Then by summing over all $\sigma\in\Ac_t$ we have
\begin{equation}\label{delta1}
    \sum_{\sigma\in\Ac_t}\tilde{\pi}^\sigma_t = \sum_{\sigma\in\Ac_t}\pi^\sigma_t + \sum_{\sigma\in\Ac_t}\delta_t^\sigma.
\end{equation} Note that $\sum_{\sigma\in\Ac_t}\tilde{\pi}^\sigma_t = 1$ by definition of $\tilde{\pi}_t$ and $\sum_{\sigma\in\Ac_t}\pi^\sigma_t = 1-\sum_{\tau\in\Dc_t}\pi^\tau_t$ since $\Ac_t$ and $\Dc_t$ partition $[d]$. Thus Equation \ref{delta1} can be reformulated as 
\begin{equation}\label{deltaspis}
    \sum_{\sigma\in\Ac_t}\delta_t^\sigma = \sum_{\tau\in\Dc_t}\pi^\tau_t.
\end{equation}
Thus
\begin{align*}\textbf{E}\left[ \mathbf{1}^{\mathsf{TM}}(i) \right] &= \sum_{\sigma \in \Xc_t} {p}_t^\sigma&\text{By Defn. of } p_t^\sigma\\
&=\sum_{\sigma \in \Xc_t} \tilde{\pi}_t^\sigma &\text{by Lemma }\ref{lem:tminvariant}\\
&=\sum_{\sigma \in \Xc_t} \delta^\sigma_t + \sum_{\sigma \in \Xc_t} \pi_t^\sigma &\text{By Defn. of } \delta_t^\sigma\\
&\leq\sum_{\sigma \in \mathcal{A}_t} \delta_t^\sigma + \sum_{\sigma \in \Xc_t} \pi_t^\sigma &\text{Since } \Xc_t\subseteq\Ac_t\\
&=\sum_{\sigma \in \mathcal{D}_t} \pi_t^\sigma + \sum_{\sigma \in \Xc_t} \pi_t^\sigma & \text{by Eq. } \ref{deltaspis}\\
&= \sum_{\sigma} c_t^\sigma \pi_t^\sigma& \text{By Defn. of } c_t^\sigma\\
&\le (1 + \epsilon) \min_{\sigma \in [d]} \left( \sum_{i=m+1}^{k-1} c_t^\sigma\right) + \frac{\ln d}{\epsilon} & \text{By Multiplicative Weights Analysis \cite{arora2012multiplicative}}  \\
&\le (1 + \epsilon)H + \frac{\ln(d)}{\epsilon}& \text{By considering the last alive path } T_\sigma
\end{align*}
 For the last inequality, 
note that if $T_\sigma$ is the last alive path then
 $\sum_{i=m+1}^{k-1}  c^\sigma_t \leq H$ by the definition of $H$. 
\end{proof}

\subsection{Monotonicity}
\label{subsec:treematchmonotonicity}

We show that any neighbor algorithm for online metrical search
can be extended to a monotone algorithm for online metrical matching, where a neighbor algorithm has the property that 
if it moves the car to a parking spot $s_i$ with positive probability then it must be the case that there is no in-commission parking spot on the route to $s_i$.
As \TS\ is obviously a neighbor algorithm, it then 
follows that it can be extended to a monotone algorithm
for online metrical matching, which we will call \TM.

\begin{lemma} 
\label{lemma:montonicitycondition}
Let $A$ be a neighbor algorithm for online metrical search. Then  there exists a monotone algorithm $B$ for online metrical matching on a tree metric that is identical to $A$ for online metrical search instances.
\end{lemma}
\begin{proof}
We construct $B$ using $A$ as a subroutine. 
As long as $B$ continues to receive requests for which there
is a co-located available server, it will send that request
to a co-located server. By renumber the requests, let $r_1$ for the first request for which there is not a co-located server.
We now know that at this point, there optimal matching has positive cost. Now $B$ starts running $A$, with the current server locations and the car parked at $r_1$. 
As long as there is an optimal matching with only one positive edge, $B$ will continue to run $A$. Whenever $A$ doesn't move the car, the request arrived at a co-located server, which is
where $B$ will send the request. If a request arrives at the location of the car in $A$, then $B$ moves the request to where $A$ moves the car. 

If a request $r_t$ arrives at a vertex without a co-located server and that is not the location of the car, then let
us call it a deviating request. $B$
handles a deviating request in the following way.
 Let $c$ be the current parking spot of the car. Let $Q$ be the collection of
in-commission parking spots that can be reached from $c$ without passing over another in-commission parking spot. 
For a parking spot $s_i \in Q$ let $p_i$ be the probability that $A$ moves the car to $s_i$ if $c$ is decommissioned.
Let $R$ be the vertices that can reach $c$ without passing
through a vertex in $Q$. Note $R$ includes no vertices in $Q$. Let $X= V- R - Q$, the vertices separated from $c$ by $Q$.
 If a request arrives at vertex  $v \in X$ then the 
 $B$ moves the request to the first in-commission parking spot on the path from $v$ to $c$. If a request arrives at a vertex $v \in Q$ then $B$ moves the request to a co-located server. If a request arrives at a vertex $v \in R$ then moves the request
 to each server $s_i \in Q$ with probability $p_i$.
 It is clear that $B$ is monotone for a deviating request $r_t$.  After receiving a deviating request, then $B$ knows that the instance is not an online metrical searching instance, and 
 can then mimic any monotone algorithm, for exactly the greedy algorithm that moves a request to the nearest available server.
\end{proof}

\section{The \HTM\ Algorithm}\label{sec:clustering}

In subsection \ref{subsec:htbuild} we describe 
an algorithm \HTB\ that 
builds a grove $G$ from a tree metric $T$
with distance metric $d_T$ before
any request arrives. We assume without loss of generality that
the minimum distance in $T$ is 1. 
In subsection \ref{subsec:htmdescription}
we then give an algorithm \HTM\ for online metrical matching on a tree metric that utilizes the
algorithm \TM\ on each tree in the grove constructed  by \HTB,
and we prove some basic properties of the grove $G$.
In subsection \ref{subsec:grovematchanalysis}
we show that \HTM\ is a monotone online metrical matching algorithm on a tree metric, and is 
$O(\log^6 \Delta \log^2 n)$-competitive for online metrical search instances.

\subsection{The \HTB\ Algorithm}
\label{subsec:htbuild}

\begin{definition}\label{def:hierarchicaltree}
A grove $G$ is either: a rooted tree $X$ consisting of a single vertex, or
an unweighted 
rooted tree $X$ with a  grove $X(v)$ associated with 
each vertex $v \in X$. The tree $X$ is  the canopy of the grove $G$.
Each $X(v)$ is  a subgrove of $X$. The canopy of a subtree $X(v)$ is a child of $X$. Trees in $G$ are descendants of $X$.
\end{definition}

\medskip
\noindent
{\bf \HTB\ Description:} \HTB\ is a recursive algorithm 
that takes as input a tree metric $T$, a designated
root $\rho$ of $T$, positive real 
$R$, a positive real $\alpha$ and a positive integer $d$. In the initial call to  \HTB,
$T$ is the original tree metric, $\rho$ is an arbitrary vertex in $T$, $R$ is the maximum distance $\Delta$ between $\rho$ and any other vertex in $T$, $d$ is $1$, and $\alpha$ is a parameter to be determined later in the analysis. 

If $T$ consists of a single vertex $v$, then the recursion ends and the algorithm outputs a rooted tree consisting of only
the vertex $v$. We call this tree a leaf of the grove. Otherwise the algorithm's first goal is to 
partition the vertices of $T$
into parts $P_1, \ldots, P_k$,
and designate one vertex $\ell_i$
of each partition $P_i$ 
as being the leader of $P_i$. To accomplish this,
the algorithm sets
partition $P_1$ to consist of the vertices in $T$ that 
are within a distance $z$ of $\rho$,
where $z$ is selected uniformly at random from
the range $[0, \frac{R}{\alpha}]$. The leader $\ell_1 $
is set to be $\rho$.
To compute $P_i$ and $\ell_i$ after the first $i-1$ parts
and leaders are computed the algorithm takes the following
steps. Let $\ell_i$ be a vertex such that  $\ell_i \notin \cup_{j=1}^{i-1} P_j$ and
for each vertex $v$ on the path $(\ell_i, \rho)$ it is the case that $v \in \cup_{j=1}^{i-1} P_j$. So $\ell_i$ is not in
but adjacent to the previous partitions.
Then $P_i$ consists of all vertices  $v \in T - \cup_{j=1}^{i-1} P_j$ 
that are within
distance $\frac{R}{\alpha}$ from $\ell_i$ in $T$. 
So $P_i$ intuitively is composed of vertices that are not
in previous partitions and that are close to $\ell_i$.

The tree $X$ at this point in the recursion  has
a vertex for each part in the partition of $T$.
There is an edge between vertices/parts $P_i$
and $P_j$ in $X$ if and only if there is
an edge $(v, w)$ in $T$ such that
$v \in P_i$ and 
$w \in P_j$. We identify this edge in $X$ with the edge
$(v, w)\in T$. 
The root of $X$ is the vertex/part $P_1$. 
The tree $X$ is at depth $d$ in the grove. 
The grove $X(P_i)$  associated with vertex $P_i$ in $X$ is the result of
calling \HTB\ on the subtree of $T$ induced by
the vertices in $P_i$, with $\ell_i$ designated
as the root, parameter $R$ decreased by
an $\alpha$ factor, parameter $\alpha$ unchanged, and 
parameter $d$ incremented by 1.

So from here on, let $G$ denote the grove built by 
\HTB\ on the original tree metric $T$.

\begin{definition} ~
\begin{itemize}
\item
For an edge  $(u,v) \in T$, let $\delta(u, v)$ be the depth in the grove $G$ of the tree $X$ 
that contains $(u, v)$. Note that each edge in $T$ occurs in exactly one tree in $G$. 
\item
For an edge $(u,v) \in T$,
define $d_{G}(u, v)$ to be $\frac{\Delta}{\alpha^{\delta(u, v)-1}}$. 
\item
For vertices $u_0, u_h \in T$, connected by
the simple path $(u_0, u_1, \dots, u_h)$ in $T$,
define $d_G(u_0, u_h)$ to be
$\sum_{i = 0}^{h-1} d_{G}(u_i,u_{i+1})$. 
Obviously $d_G$ forms a metric on the vertices of $T$.  
 \end{itemize}
 \end{definition}
 
\begin{lemma}
\label{usefulFacts} 
 Recall that $d_T(u, v)$ is the shortest path distance between two vertices $u, v$ of tree $T$. For all vertices $u, v\in T$, we have that $d_{G}(u, v)\geq d_T(u, v)$
        and $\Exp{d_{G} (u, v)} \leq \alpha(1+\log \Delta)  \cdot d_T(u, v)$.
\end{lemma}
\begin{proof}
Let $(u_0, u_1, \dots, u_h)$ be the path from $u_0$ to $u_h$
  in $T$. Since   $d_{G}(u_0,u_h) = \sum_{i = 0}^{h-1} d_{G}(u_i,u_{i+1})$, and $d_{T}(u_0,u_h) = \sum_{i = 0}^{h-1} d_{T}(u_i,u_{i+1})$, it is sufficient to prove this for each $(u, v) \in T$. For notational simplicity let $\delta := \delta(u,v)$. 
  
Notice that from the construction of $G$ that $d_G(u,v) = \frac{\Delta}{\alpha^{\delta-1}}$. Showing that $d_T(u, v) \leq \frac{\Delta}{\alpha^{\delta-1}}$ proves the first inequality: if $\delta = 1$ then $d_T(u, v) \leq \Delta$, else, $d_T(u, v) \leq \frac{\Delta}{\alpha^{\delta - 1}}$ since $(u, v)$ did not get cut at depth $\delta - 1$. 

To prove the second inequality, let $A_i$ be the event that 
$\delta = i$, and $A_{< i}$ be the event that $\delta < i$. Notice that if $d_T(u, v) \geq \frac{\Delta}{\alpha^{i}}$ for some $i$, then, $\delta \leq i$ since cuts at depth $i$ of the recursion are made in increments of $\frac{\Delta}{\alpha^i}$ distance. Hence, the value of $\delta$ is at most $1 + \log_\alpha(\frac{\Delta}{d_T(u,v)})$. By the linearity of expectation we have:
  \begin{align*}
  \Exp{d_{G}(u,v)} &= \sum_{i = 1}^{1 + \log_\alpha(\Delta /d_T(u,v))} \Prob{A_i} \cdot \frac{\Delta}{\alpha^{i-1}} = \sum_{i = 1}^{1 + \log_\alpha(\Delta/d_T(u,v))} \Prob{A_i |\;  \overline{A_{< i}}\; } \Prob{\; \overline{A_{< i}}\;} \cdot \frac{\Delta}{\alpha^{i-1}} \\ 
  &\leq \sum_{i = 1}^{1 + \log_\alpha(\Delta/d_T(u,v))} \Prob{A_i |\;  \overline{A_{< i}}\; }\cdot \frac{\Delta}{\alpha^{i-1}} = \sum_{i = 1}^{1 + \log_\alpha(\Delta/d_T(u,v))} \alpha d_T(u,v).
  \end{align*} 
  The last equality follows from the fact that  $\Prob{A_i | \overline{A_{< i}}\;} = \frac{d_T(u, v)}{\Delta/\alpha^i}$ 
  since cuts at depth $i$ of the recursion are made in increments of 
  $\frac{\Delta}{\alpha^i}$ distance with an offset randomly
  chosen from $[0, \frac{\Delta}{\alpha^i}]$.  
\end{proof}

\begin{comment}
So roughly speaking the root of the heirarchical tree generated by \HTB is a tree where the cost of each edge is the diameter of the
metric space divided by a parameter $\alpha$. Level $i$ of a grove consists of a collection of trees, where there is a bijection between
each tree on this level and the nodes in the next higher level of the grove. Further the cost of each edge on this level is the diameter
of the tree metric space divided by $\alpha^i$. 
As in an HST, there is a bijection between the vertices of the original metric space and the leaves of the grove.
\end{comment}

\begin{corollary}\label{HTreduc}
 An algorithm \textbf{B} that is $c$-competitive 
 for online metric matching on $T$ 
 with distance  metric $d_{G}$  is $O(c \cdot \alpha  \log  \Delta)$-competitive for online metric matching on $T$ with distance metric $d_T$. 
 \end{corollary}
 \begin{proof}
 This is an immediate consequence of Lemma \ref{usefulFacts}.
 \end{proof}

\subsection{\HTM\ Description}
\label{subsec:htmdescription}
We now describe an  algorithm \HTM\ for online metrical
matching for tree metrics.

\medskip
\noindent
{\bf \HTM\ Description:} 
Conceptually within \HTM, 
a separate copy \TM($X$) of the online metric matching algorithm
\TM\ will be run on each tree $X$ in the grove $G$ constructed by the algorithm \HTB.
In order to accomplish this, we need to initially 
place servers
at the vertices  in $X$. We set the number of
servers initially located at each vertex $x \in X$  to the number of servers in $T$ that are located at vertices
$v \in T$ such that 
 $v \in x$ (recall that each vertex in a tree in the grove  $G$ corresponds to a collection of vertices in $T$). 

When a request $r_t$ arrives at a vertex $v$ in
$T$, 
the algorithm \HTM\ calls the algorithm \TM\ 
on a sequence $(X_1, x_1), (X_2, x_2), \ldots$
where each $X_i$ is a tree of depth $i$ in $G$
and $x_i$ is a vertex in $X_i$. 
Initially $X_1$ is the depth 1 tree in $G$,
and $x_1$ is the vertex in $X_1$ that contains
$v$. Assume that \TM\ has already been called on
 $(X_1, x_1), (X_2, x_2), \ldots (X_{i-1}, x_{i-1})$,
then the algorithm \HTM\ processes
 $(X_i, x_i)$ in the following manner.
 First, \TM($X_i$) is called to respond to a request at $x_i$.
 Let $y_i$ be the vertex in $X_i$ that \TM($X_i$) moved this request to. 
If $X_i$ is a leaf in $G$, then \TM($X_i$) sets $y_i = x_i$,
and \HTM\ moves request $r_t$  to the unique vertex in $T$ corresponding
to $x_i$.  
If $X_i$ is not a leaf in $G$, then
$X_{i+1}$ is set to be the canopy of the grove
$X_i(y_i)$, and  
$x_{i+1}  = \argmin_{w \in T  : w \in X_{i+1}} d_T(v, w)$
or equivalently $x_{i+1}$ is the first
vertex in $X_{i+1}$ that one encounters if one walks in $T$ from
$v$ to the vertices of $X_{i+1}$.

\begin{lemma}\label{hopdist}
Consider a tree $X$ at depth $\delta$ with root $\rho$ in grove $G$. For any vertex $v$ in $X$, the number of hops in $X$ between $\rho$ and $v$ is at most  $\alpha+1$. Furthermore, by the time that  \TM($X$)   enters its core phase,  it must be the case that for every descendent tree $Y$ of $X$ in $G$
there will be no future movement of the car on edges in $Y$ 
while \TM($Y$)  is in its prologue phase. 
\end{lemma}
\begin{proof}
This follows immediately from the fact that the parameter
$R$ decreases by an $\alpha$ factor on each recursion.
\end{proof}

\subsection{\HTM\ Analysis}
\label{subsec:grovematchanalysis}

We now analyze \HTB\ and \HTM\ under the assumption that $\alpha = (\ln n)(\log_\alpha^2\Delta)$ and $\epsilon = \frac{1}{\log_\alpha \Delta}$.

\begin{lemma}\label{lem: HTMGood}
The algorithm \HTM\ is $O(\log n\log^3 \Delta)$-competitive for online metrical search instances with the metric $d_G$.
\end{lemma} 
\begin{proof}
If \HTM\ directs a request to traverse an edge $(u, v) \in T$, we will say that the cost of this traversal is
charged to the unique tree in $G$ that contains $(u, v)$. 
Define $P(\delta)$ to be the charge incurred by a tree $X$
of depth $\delta$ in $G$ and all subgroves $X(v)$
of $X$ during the prologue phase  of \TM($X$). 
Define $C(\delta)$ to be the charge incurred by a tree $X$
of depth $\delta$ in $G$ and all subgroves $X(v)$
of $X$ during the core phase of  \TM($X$). 

Recall that the distance under the $d_G$ metric
 of ever edge in $X$ is $\frac{\Delta}{\alpha^{\delta-1}}$ and by Lemma \ref{hopdist} there are at most $\alpha+1$ vertices on the path from any leaf to the root of $X$. This gives us that the distance in $X$ under $d_G$ from the root to any leaf is at most $\alpha\frac{\Delta}{\alpha^{\delta-1}}=\frac{\Delta}{\alpha^{\delta-2}}$ and that the diameter of $X$ is at most $2\frac{\Delta}{\alpha^{\delta-2}}$. 
The only subgroves $X(v)$ of $X$ that incur costs during the prologue
phase of \TM($X$)   are those subgroves
for which $v$ is traversed by the car on its path to the root of
$X$. Thus we obtain the following recurrence:

\begin{equation}\label{Pdelta}
    P(\delta)\leq (\alpha+1)\left(P(\delta+1) + C(\delta+1)\right) +  \frac{\Delta}{\alpha^{\delta-2}}.
\end{equation} 
Note that once the core phase begins in \TM($X$), 
by Lemma \ref{hopdist} all instances of \TM($Y$) on any tree $Y$ that is a descendent of $X$ in $G$ can incur no most costs in their prologue phase. 
By Theorem \ref{thm:MSUnit}  the core phase cost on $X$
is at most $(1+\epsilon)(\alpha+1)+\frac{\ln n}{\epsilon}$ times 
the diameter of $X$, which is at most  $2\frac{\Delta}{\alpha^{\delta-2}}$. Thus we obtain the following
recurrence:
\begin{equation}\label{Cdelta}
    C(\delta)\leq\left(C(\delta+1)+2\frac{\Delta}{\alpha^{\delta-2}}\right)\left((1+\epsilon)(\alpha+1)+\frac{\ln n}{\epsilon}\right)\\
\end{equation}

We expand the recurrence relation for $C(\delta)$ first. Treating $((1+\epsilon)(\alpha+1)+\frac{\ln(n)}{\epsilon})$ as a constant $Z$, and expanding $C(\delta)$ we obtain: 
\begin{align*}
    C(\delta)&\leq \left(C(\delta+1)+2\frac{\Delta}{\alpha^{\delta-2}}\right)Z\\
    &= C(\delta+1)Z+2\frac{\Delta}{\alpha^{\delta-2}}Z\\
    &\leq 2\frac{\Delta}{\alpha^{\delta-1}}\sum_{i=1}^{\log_\alpha \Delta}\left(\frac{Z}{\alpha}\right)^i\\
    &\leq \frac{2\Delta\log_\alpha \Delta  }{\alpha^{\delta-1}} \left(\frac{Z}{\alpha }\right)^{\log_\alpha(\Delta)}\\
    &\leq\frac{2\Delta\log_\alpha \Delta  }{\alpha^{\delta-1}}\left((1+\epsilon)\frac{\alpha+1}{\alpha}+\frac{\ln(n)}{\epsilon\alpha}
    \right)^{\log_\alpha(\Delta)}\\
    &\leq \frac{2\Delta\log_\alpha \Delta  }{\alpha^{\delta-1}}\left(1 + \epsilon + \frac{1}{\alpha}+\frac{\epsilon}{\alpha}+\frac{\ln n}{\epsilon\alpha}\right)^{\log_\alpha(\Delta)}\\
    &\leq \frac{2\Delta\log_\alpha \Delta  }{\alpha^{\delta-1}}\left(1 + \frac{1}{\log_\alpha \Delta} + \frac{1}{(\ln n)(\log_\alpha^2\Delta)}+\frac{1}{(\ln n)(\log_\alpha^3\Delta)}+\frac{1}{\log_\alpha \Delta }\right)^{\log_\alpha(\Delta)}\\
    &\leq \frac{2\Delta\log_\alpha \Delta  }{\alpha^{\delta-1}}\left(1 + \frac{4}{\log_\alpha \Delta}\right)^{\log_\alpha \Delta}\\
    &\leq\frac{2e^4 \Delta\log_\alpha \Delta  }{\alpha^{\delta-1}}
\end{align*}
Now expanding the recurrence relation for $P(\delta)$ we obtain:
\begingroup
\allowdisplaybreaks
\begin{align*}
    P(\delta)&\leq (\alpha+1)\left(P(\delta+1) + C(\delta+1)\right) + \frac{\Delta}{\alpha^{\delta-2}}\\
    &\leq (\alpha+1)\left(P(\delta+1)+\frac{2 e^ 4\Delta}{\alpha^{\delta -1}}\log_\alpha \Delta\right)+\frac{\Delta}{\alpha^{\delta-2}}\\
    &= (\alpha+1)P(\delta+1)+\frac{2 e^ 4\Delta (\alpha + 1)}{\alpha^{\delta -1}}\log_\alpha \Delta+\frac{\Delta}{\alpha^{\delta-2}}\\
    &\le (\alpha+1)P(\delta+1)+\frac{2 e^ 4\Delta (\alpha + 1)}{\alpha^{\delta -1}}\log_\alpha \Delta+\frac{\Delta (\alpha + 1)}{\alpha^{\delta-1}}\\
    &\le (\alpha+1)P(\delta+1)+\frac{\Delta (\alpha + 1)}{\alpha^{\delta -1}}\left( 2e^4 \log_\alpha \Delta + 1\right)\\
    &\leq(\alpha+1)P\left(\delta+1\right) + \frac{3e^4(\alpha+1) \Delta \log_\alpha \Delta}{\alpha^{\delta-1}}\\
    &\leq\frac{3e^4 \Delta \log_\alpha \Delta}{\alpha^{\delta-2}}\sum_{i=1}^{\log_\alpha \Delta }\left(\frac{\alpha+1}{\alpha}\right)^i\\
    &\leq \frac{3e^4 \Delta \log^2_\alpha \Delta}{\alpha^{\delta-2}}\left(\frac{\alpha+1}{\alpha}\right)^{\log_\alpha \Delta}\\
     &= \frac{3e^4 \Delta \log^2_\alpha \Delta}{\alpha^{\delta-2}}\left(1 + \frac{1}{\alpha}\right)^{\log_\alpha \Delta}\\
      &= \frac{3e^4 \Delta \log^2_\alpha \Delta}{\alpha^{\delta-2}}\left(1 + \frac{1}{(\ln n)(\log_\alpha^2\Delta)}\right)^{\log_\alpha \Delta}\\
    &\leq \frac{3e^5 \Delta \log^2_\alpha \Delta}{\alpha^{\delta-2}}
\end{align*}
\endgroup
Hence the cost of the algorithm \HTM\ is $O\left(\frac{\Delta}{\alpha^{\delta-2}}\log^2 \Delta\right)$. However, note that \TM\ only pays positive cost on $X$ if for any optimal solution there is at least one request that such a solution must pay positive cost for in $X$. The reason for this is that if \TM($X$) moves the car
out of a vertex $v$ in $X$, then there are no in-commission parking spots left in $v$, and therefore every algorithm would have
to move the car out of $v$. Since every edge in $X$ has distance $\frac{\Delta}{\alpha^{\delta-1}}$, this gives us that \HTM\ must be $O(\alpha\log^2 \Delta) = O(\log n \log^3 \Delta)$ competitive on the metric $d_{G}$. 
\end{proof}
Together with Corollary \ref{HTreduc}, Lemma \ref{lem: HTMGood} gives us the following theorem: 
\begin{theorem}
  \HTM\ is $O(\log^6 \Delta\log^2 n)$-competitive for online metrical search instances. 
\end{theorem}

\begin{lemma}\label{LemAna2} \HTM\ is a monotone  algorithm for online metrical matching.
\end{lemma}
\begin{proof} 
Consider a tree $X_i$ and vertex $x_i \in X_i$ considered in the
\HTM\ algorithm. Then by the monotonicity of the algorithm
\TM\ the probability that a request would arrive
at a particular $y_i$ if the request arrived in a vertex $x_j$ on the path from
$x_i$ to $y_i$ has to be at least the probability 
that a request arriving at $x_i$ moves to $y_i$. 
And the routing in the $X_k$'s, $k > i$, is independent
of whether the request arrived at $x_i$ or $x_j$. 
\end{proof}

\section*{Acknowledgements} \vspace{-.1 in}We thank Anupam Gupta for his  guidance, throughout the research process, that was 
absolutely critical to obtaining these results.  
We thank Amos Fiat for introducing us to this posted-price research area, and for several helpful discussions.

\bibliography{arxiv}

\end{document}